\documentclass[journal,draftcls,onecolumn,12pt,twoside]{IEEEtranTCOM}

\usepackage{exscale}
\usepackage{amsmath,esint}
\usepackage{amssymb}
\usepackage{graphicx,float}
\usepackage{fixltx2e}
\usepackage{subcaption}
\usepackage{epstopdf}
\usepackage{multirow}
\usepackage{array}
\usepackage{cite}
\usepackage{color}

\DeclareMathOperator*{\vc}{vec}
 
\newcommand{\bx}{\mathbf{x}}    
\newcommand{\by}{\mathbf{y}}    
\newcommand{\bw}{\mathbf{w}}    
\newcommand{\bg}{\mathbf{g}}    
\newcommand{\bH}{\mathbf{H}}    
\newcommand{\be}{\mathbf{e}}    
\newcommand{\bh}{\mathbf{h}}    
\newcommand{\bN}{\mathbf{N}}    
\newcommand{\bC}{\mathbf{C}}    
\newcommand{\bD}{\mathbf{D}}    
\newcommand{\bP}{\mathbf{P}}    
\newcommand{\bI}{\mathbf{I}}    
\newcommand{\bV}{\mathbf{V}}    
\newcommand{\bu}{\mathbf{u}}    
\newcommand{\bS}{\mathbf{S}}    
\newcommand{\bT}{\mathbf{T}}    
\newcommand{\bX}{\mathbf{X}}    
\newcommand{\bY}{\mathbf{Y}}    
\newcommand{\bU}{\mathbf{U}}    
\newcommand{\bA}{\mathbf{A}}    
\newcommand{\bB}{\mathbf{B}}    
\newcommand{\define}{\triangleq}    

\newcommand{\mbbC}{\mathbb{C}}	
\newcommand{\mbbR}{\mathbb{R}}	
\newcommand{\mccn}{\mathcal{CN}}	

\newcommand{\btt}{\bsym{\theta}}

\newcommand{\bzro}{\mathbf{0}}

\newcommand{\ben}{\begin{enumerate}} 	  	
  \newcommand{\een}{\end{enumerate}} 			
\newcommand{\beq}{\begin{equation}} 	  	
\newcommand{\eeq}{\end{equation}} 			
\newcommand{\bes}{\begin{equation*}}
\newcommand{\ees}{\end{equation*}}

\newcommand{\bea}{\begin{eqnarray}}		
\newcommand{\eea}{\end{eqnarray}} 		
\newcommand{\beas}{\begin{eqnarray*}}
  \newcommand{\eeas}{\end{eqnarray*}}
\newcommand{\ba}{\begin{array}}
  \newcommand{\ea}{\end{array}}
\newcommand{\sbea}{\nopagebreak[3]\samepage\begin{eqnarray}}
\newcommand{\seea}{\end{eqnarray}\pagebreak[0]}
\newcommand{\sbeas}{\nopagebreak[3]\samepage\begin{eqnarray*}}
  \newcommand{\seeas}{\end{eqnarray*}\pagebreak[0]}

\newcommand{\lb}{\label}
\newcommand{\er}[1]{{\rm(\ref{#1})}}

\newcommand{\bit}{\begin{itemize}}
  \newcommand{\eit}{\end{itemize}}
\newcommand{\bsym}{\boldsymbol}

\newcommand{\ti}{\textit}
\newcommand{\nn}{\nonumber}

\DeclareMathOperator{\Trace}{Tr}
\DeclareMathOperator{\diag}{diag}
\DeclareMathOperator{\Real}{Re}

\newtheorem{theorem}{Theorem}
\newtheorem{lemma}{Lemma}

\newtheorem{corollary}{Corollary}

\newenvironment{proof}[1][Proof]{\begin{trivlist}
    \item[\hskip \labelsep {\bfseries #1}]}{\end{trivlist}}

\newcommand\blfootnote[1]{%
  \begingroup
  \renewcommand\thefootnote{}\footnote{#1}%
  \addtocounter{footnote}{-1}%
  \endgroup
}

\newcommand{\qed}{\nobreak \ifvmode \relax \else
  \ifdim\lastskip<1.5em \hskip-\lastskip
  \hskip1.5em plus0em minus0.5em \fi \nobreak
  \vrule height0.75em width0.5em depth0.25em\fi}

\def\figurename{Fig.}
\def\iflatex{\iftrue}
\def\ifcomments{\iffalse}
\begin{document} 
  \title{PCA-based Antenna Impedance Estimation in Rayleigh Fading Channels}
  \author{Shaohan~Wu~and~Brian~L.~Hughes
  }
  
  \maketitle
  \begin{abstract}
   Impedance matching between receive antenna and front-end significantly impacts channel capacity in wireless channels. To implement capacity-optimal matching, the receiver must know the antenna impedance. But oftentimes this impedance varies with loading conditions in the antenna near-field. To mitigate this variation, several authors have proposed antenna impedance estimation techniques. However, the optimality of these techniques remains unclear.  In this paper, we consider antenna impedance estimation at MISO receivers over correlated Rayleigh fading channels. 
   We derive in closed-form the optimal ML estimator  under i.i.d. fading and then show it can be found via a scalar optimization in generally correlated fading channels.  Numerical results suggest a computationally efficient,  principal-components approach that estimates antenna impedance in real-time for all Rayleigh fading channels. Furthermore, ergodic capacity can be significantly boosted at originally poorly matched receivers with adaptive matching using our proposed approach. 
   \blfootnote{S. Wu is with MediaTek USA, Irvine CA, 92606 (e-mail: swu10@ncsu.edu). B. L. Hughes is with the Department of Electrical and Computer Engineering, North Carolina State University, Raleigh, NC 27695 (e-mail: blhughes@ncsu.edu).}
  \end{abstract}
  

  \section{Introduction}
  
Over the past two decades, several authors have demonstrated that antenna impedance matching at the receiver can significantly impact capacity and diversity in wireless communication channels \cite{domi,domi2,lau,wall,gans,gans2}. For multiple-input, multiple-output (MIMO) channels, multi-port matching techniques that optimize capacity for different front-end configurations have been investigated in \cite{domi,domi2,lau,wall,gans,gans2}. These works show that capacity is sensitive to receiver impedance matching, and optimal matching can dramatically increase the capacity of wireless channels. 

To implement the capacity-optimal matching in \cite{domi,domi2,lau,wall,gans,gans2}, the receiver must know the antenna input impedance. In practice, this is complicated by the fact that this impedance varies with loading conditions in the antenna near-field. For example, the position of a hand on a cellular handset can significantly affect the input-impedance of handset antennas \cite{ali,moha,vasi}. For this reason, several authors have proposed the use of adaptive matching networks that estimate and adapt to variations in antenna impedance. In \cite{vasi,moha}, the authors propose adaptive matching techniques based on empirical capacity estimates at the receiver. Capacity is optimized by periodically adjusting the receiver impedance, and calculating the impact on empirical capacity estimates. Capacity is then optimized by searching for the receiver impedance that maximizes the observed capacity. This approach has several advantages: it is general, requires only information that is easily available at the receiver, and numerical results suggest it converges to the optimal capacity in a matter of seconds. However, because this approach requires a search of the receiver impedance space, it is also computationally-intensive, and convergence is slow compared to other common communication tasks, such as channel estimation, which often converges in milliseconds.

In this paper, we explore a new, more direct, approach to capacity-optimal adaptive matching, based on direct estimation of the receiver antenna impedance $Z_A$ at multiple-input, single-output (MISO) receivers. This approach seeks to first estimate $Z_A$ at the receiver, and then use the analytical results of \cite{domi,domi2,lau,wall,gans,gans2} to calculate the capacity-optimal matching for this estimate. In contrast to the empirical capacity metric considered in \cite{vasi,moha}, $Z_A$ does not depend on the transmitted data, and thus should be easier to estimate. More importantly, since this approach {\em calculates} the optimal matching network, rather than searching over all possible networks, it should use far less data and computation than in \cite{vasi,moha}, and so has the potential to optimize capacity more quickly.

Most current wireless receivers have no mechanism to estimate antenna impedance. However,
they do have extensive resources for channel estimation, in the form of training sequences and pilot symbols. These resources are often underutilized, in the sense that they are designed for worst-case conditions which rarely occur (e.g., high-speed trains in dense urban multi-path). In this work, we consider diverting some of these resources to estimate the receiver antenna impedance in addition to the channel path gains. We consider an estimation approach in which the receiver perturbs its impedance during reception of a known training sequence. Using these observations, the receiver then performs joint channel and antenna impedance estimation based on the received data. Finally, the impedance estimates are used by the receiver to adaptively adjust the antenna matching network in order to maximize the resulting ergodic capacity, using the results of \cite{domi,domi2,lau,wall,gans,gans2}.


Specifically, we develop a classical  framework of antenna impedance estimation at MISO receivers in Raleigh fading channels. 
Based on observation of training sequences via synchronously switched load at the receiver, we derive the true  maximum-likelihood (ML) estimators for antenna impedance, treating the fading path gains as nuisance parameters. This ML estimator is derived over multiple packets in general under correlated Rayleigh fading channel. 
For a single-packet and/or a fast fading channel, this ML estimator is derived in closed-form based on the eigenvector corresponding to the largest eigenvalue of the sample covariance matrix. When the channel is temporally correlated, the ML estimator can be found via a scalar optimization. 
We explore the performance (e.g., efficiency) of these estimators through numerical examples. The impact of channel correlation on impedance estimation accuracy is also investigated.

Next, we review recent works that are most relevant to this paper, and elaborate our contributions by comparing our approach to theirs. 

%
%



Recently, several studies have considered direct estimation of the receiver antenna impedance\cite{ali,hass,wu,wu2}. In \cite{ali}, Ali \ti{et al} propose an approach to measure the antenna impedance that compares the received signal power at different frequencies and different loads.
Most similar to this paper is the work of Hassan and Wittneben \cite{hass}, who have considered joint channel and  impedance estimation for MIMO receivers. In their work, the authors vary the receiver load impedance and rely on known training sequences to jointly estimate the channel and the antenna impedance.

Similar to this paper, \cite{ali} and \cite{hass} both share the goal of estimating the antenna impedance directly, and both involve changing the receiver load to do so. However, these works differ significantly from this paper in modeling, technical approach and performance metrics.
Ali \ti{et al} assume a deterministic receiver circuit model that directly observes power, whereas we assume a noisy receiver which observes only the demodulated and sampled output of the receiver front-end. Their approach involves solving deterministic nonlinear equations via simulations and measurements, whereas our approach is grounded in estimation theory. 
In this work, we also  perturb the receiver's load impedance. However, rather than using special circuits to measure power, we estimate $Z_A$ by applying  estimation theory to observations based on known training sequences. This approach has several important advantages: (1)
it uses resources which are already present in most wireless systems;
(2) no additional measurement circuits are needed, except the ability to perturb the receiver's load; and
(3) estimation theory provides tools to evaluate how performance depends on the training sequences and receiver perturbations, which may enable optimization over these parameters.
The results of \cite{hass} also differ from ours in several important ways: First, the authors use a different model of noise in the receiver front end, which includes both antenna and load resistor noise, whereas we consider a scenario in which amplifier noise dominates. Second, \cite{hass} considers only uncorrelated, fast-fading conditions, whereas we consider the more general and realistic case of spatially and temporally-correlated fading. Third,
\cite{hass} considers estimation of the channel fading path gains and antenna impedance, whereas we consider a different parameterization of the channel, that estimates quantities that are more relevant to those needed by the communication algorithms, and which leads to better behaved estimators. Fourth, \cite{hass} considers least-squares estimation, whereas we consider maximum-likelihood estimation as well as fundamental lower bounds on performance, such as the Cram\'er-Rao bound (CRB). Finally, through numerical examples, we also
investigate the impact the channel and impedance estimation error on the ergodic capacity and on the receiver signal-to-noise ratio when used in conjunction with adaptive matching.
In our previous works, joint channel and antenna impedance estimators of single-antenna receivers are derived in classical estimation \cite{wu} and hybrid settings \cite{wu2}. However, neither of these two settings results in the optimal impedance estimator \cite{mes09}. In this paper, we formulate the same impedance estimation problem using the marginal probability density function (PDF), and then derive the optimal impedance estimator by treating channel fading path gains as nuisance parameters.

The rest of the paper is organized as follows. We present the system model in Sec.~\ref{3secII}, and derive the true ML estimators for the antenna impedance with one packet in Sec.~\ref{3secIII}.
We then investigate the more general problem of multiple packets, and derive the true ML estimator and a simple method of moments (MM) estimation in Sec.~\ref{3secIV}. We explore the performance of these estimators through numerical examples in Sec.~\ref{3secV}, and  conclude  in Sec.~\ref{3secVI}.

\section{System Model}
\lb{3secII}
Consider a narrow-band, multiple-input, single-output (MISO) channel with $N$ transmit antennas and one receive antenna. A circuit model of the receiver is illustrated in Fig.~\ref{fig_SystemModel}, which models a scenario in which amplifier noise dominates \cite{gans2,lau,wall}. 
\begin{figure}[t!]
	\begin{center}
		\includegraphics[width=.4\textwidth, keepaspectratio=true]{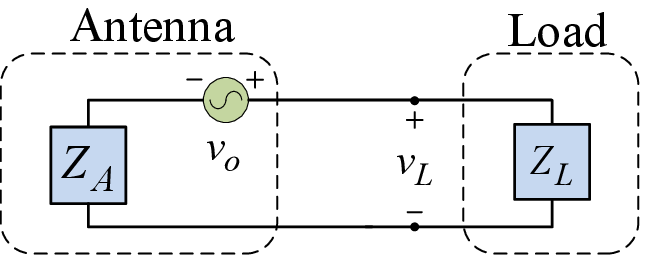}
	\end{center}
	\vspace{-10pt}
	\caption{Circuit model of a single-antenna receiver}
	\vspace{-18pt}
	\label{fig_SystemModel}
\end{figure}
In Fig.~\ref{fig_SystemModel}, the antenna is modeled by its Thevenin equivalent
\beq
v \ = \ Z_A i + v_o \ , \label{antenna}
\eeq
where $v, i \in \mathbb{C}$ are the voltage across, and current into, the antenna terminals. The antenna impedance is
\beq
Z_A \ = \ R_A + j X_A \ ,
\eeq
where $R_A$ and $X_A$ are the resistance and reactance, respectively. In \er{antenna}, $v_o \in \mathbb{C}$ is the open-circuit voltage induced by the incident signal field, which can be modeled in a flat-fading environment as \cite{domi2}\cite[eq.~1]{vu},
\beq
v_o \ = \ \bg^T\bx \ ,
\eeq
where $\bx\in\mbbC^N$ is the vector of symbols sent from the $N$ transmit antennas, and $\bg\in\mbbC^N$ is a vector of channel fading path gains.  We consider
a Rayleigh fading environment in which the transmit antennas are spaced far apart, so the path gains can be modeled as independent, zero-mean Gaussian random variables,
$\bg\sim\mccn(\bzro,\sigma_g^2\bI_N)$.

We assume the estimation algorithms observe a noisy version of the load voltage in Fig.~\ref{fig_SystemModel}, given by
 \cite{gans2,lau,wall}
\beq\label{eq_signal_model_ch3}
v_L \ = \ \frac{Z_L v_o}{Z_A + Z_L}+n_L \ ,
\eeq
where the observation noise $n_L$ is a zero-mean, Gaussian random variable, $n_L \sim \mccn(0,\sigma_L^2)$. In estimation theory, the performance of estimators is usually a function of the signal-to-noise ratio (SNR), which is conventionally defined as the squared-magnitude of the signal divided by the noise variance, i.e. $|Z_L v_o|^2/(\sigma_L^2 |Z_A+Z_L|^2)$.
As noted in \cite{ivrl}, however, since circuit power depends on both voltage and current,
this SNR formula does not correctly predict the ratio of the physical signal and noise powers in the receiver front-end. For a given $v_o$, the ratio of the physical signal power to noise power at the load is given by \cite[eq.~4.65c]{poza}
\beq\label{rho_sym}
\frac{R_L |v_o|^2}{\sigma_{n}^2|Z_A+Z_L|^2} \ ,
\eeq
where $\sigma_n^2$ represents the noise power at the output of the amplifier and $R_L=\Real\{Z_L\}$ . As in \cite{ivrl}, we correct this discrepancy by defining $\sigma_L^2$ in a way that ensures the SNR and physical power ratio \er{rho_sym} coincide:
\[
\sigma_L^2 \ \define \frac{|Z_L|^2 \sigma_n^2}{R_L} \ .
\]
With this definition, it is convenient to redefine the observed signal as
\bea
u & \define & \frac{\sqrt{R_L}}{Z_L} v_L \ = \ \frac{\sqrt{R_L} v_o }{Z_A + Z_L} + n \ , \label{MISOchannel}
\eea
where $n \sim \mccn(0,\sigma_n^2)$ now represents the noise referred to the amplifier output.
The model \er{MISOchannel} now correctly connects estimator performance to the physical signal-to-noise ratio \er{rho_sym}. This connection is essential in order to accurately predict the impact of impedance mismatch on system-level performance metrics, such as capacity and channel mean-squared estimation error.

Suppose the channel gains and antenna impedance are unknown to the receiver. We would like to estimate these parameters using the observations of known training sequences. We assume the transmitter sends a known training sequence, $\bx_1 , \ldots , \bx_T\in\mbbC^N$. During transmission,  the receiver shifts synchronously through a sequence of known
impedances $Z_{L,1} , \ldots , Z_{L,T}$. If $\bg$ and $Z_A$ are fixed for the duration of the training sequence, the received observations are given by
\bea\label{load_observations}
{u}_{t} &=& \frac{\sqrt{R_{L,t}} \,\bg^T\bx_t}{Z_A + Z_{L,t}} + n_{t} \ , ~~~\ t = 1, \ldots , T \ ,
\eea
where $n_{t} \sim\mathcal{CN}(0, \sigma_{n}^2)$ are independent and identically distributed.

We consider load impedances that take on only two possible values \cite{wu},
\begin{eqnarray}\label{zL_switch}
Z_{L,t} \ = \ \left\{ \begin{array}{cc}
Z_1 , & 1 \leq t \leq K \ , \\
Z_2 , & K < t \leq T \ .
\end{array} \right.
\end{eqnarray}
where $Z_1$ and $Z_2$ are known. Again we assume $Z_L=Z_1$ is the load impedance used to receive the transmitted data, and is matched to our best estimate of $Z_A$; additionally $Z_L=Z_2$ is an impedance variation introduced in order to make $Z_A$ observable.

With this choice of load impedance, we can express the observations in a simpler, bilinear form.
Note that in order the perform optimal detection of the transmitted symbols in \er{MISOchannel}, the receiver needs an estimate of the coefficient of $\bx$,
\beq\label{h_pca}
\bh \ \define \ \frac{\sqrt{R_1}}{Z_A + Z_1}\bg \sim\mccn\left(\bzro,\sigma_h^2\bI_N\right) \ , ~~~~\sigma_h^2 \ = \ \frac{R_1\cdot\sigma_g^2}{|Z_A+Z_1|^2} \ ,
\eeq
where $R_1=\Real\{Z_1\}$. It therefore makes sense to estimate $\bh$ directly, rather than the path gains $\bg$. Expressing the observations \eqref{load_observations} in terms of $\bh$,
we obtain the bilinear model
\beq\label{3observations2}
u_t  \ = \ \begin{cases}
	\bh^T \bx_t + n_t \ , & 1 \leq t \leq K \ , \\
	F \bh^T \bx_t  + n_t \ , & K < t \leq T \ ,
\end{cases}
\eeq
where we define
\beq\label{3F}
F \ \define \ \frac{\sqrt{R_2}(Z_1+Z_A)}{\sqrt{R_1}(Z_2+Z_A)} \ .
\eeq

The goal of this paper is to derive estimators for $\bh, Z_A$ and $\sigma_h^2$ based on the observations \er{3observations2}. From the invariance principle of maximum-likelihood  estimation (MLE) \cite[pg. 185]{kay}, knowing the MLE of $F$ is equivalent to knowing that of $Z_A$ and vice versa. It therefore suffices to derive estimators for $\bh, F$ and $\sigma_h^2$. In our prequel, we proposed a hybrid approach to jointly estimate $\bh$ and $F$ for the special case $N=1$ and $\sigma_h^2$ known \cite{wu2}. In this paper, we consider an alternative approach for $N \geq 1$ and $\sigma_h^2$ unknown, in which we solve this problem in two successive steps: First, we consider joint maximum-likelihood estimation of $F$ and $\sigma_h^2$, treating $\bh$ as a nuisance parameter. Second, given estimates of $F$ and $\sigma_h^2$, we then estimate $\bh$ using minimum mean-squared error estimation. Since the second step involves well-known techniques, we focus exclusively on estimators for $F$ and $\sigma_h^2$ in the next few sections; estimators for $\bh$ will be explored through numerical examples in Sec.~\ref{3secV}.

\section{Maximum-Likelihood Estimators}\label{3secIII}
In this section, we derive joint maximum-likelihood estimators for the channel variance $\sigma_h^2$ in \er{h_pca} and the antenna impedance, represented here by the parameter $F$ in \er{3F}. Before doing so, it is convenient to reduce the observations \er{3observations2} to sufficient statistics.
\begin{lemma}[Sufficient Statistics]\lb{3sslemma}  Consider the observations ${\bf u} \define (u_1, \ldots , u_T)^T$ in
\eqref{3observations2}, where $F$ and $\sigma_h^2$ are unknown constants. Suppose the matrices
\bea \label{3ss}
	\bB_1 \ \triangleq \ \sum_{t=1}^K \bx_t^* \bx_t^T   \ , \ \ \bB_2 \ \triangleq \ \sum_{t=K+1}^T  \bx_t^* \bx_t^T \ ,
	\eea
are  non-singular, where $\bx_1, \ldots , \bx_T$ is the training sequence. Then
\begin{eqnarray}\label{ys}
\by_1 \ \define \ \bB_1^{-1} \sum_{t=1}^K  u_t \bx_t^* \ , \ \ \by_2 \ = \ \bB_2^{-1} \sum_{t=K+1}^T  u_t \bx_t^* \ ,
\end{eqnarray}
are sufficient statistics to estimate $F$ and $\sigma_h^2$ based on the observation $\bu$. Moreover, $\by_1$ and $\by_2$ are conditionally independent given $\bh$, with conditional distributions
\beas
\by_1 \sim\mathcal{CN}(\bh, \sigma_{n}^2 \bB_1^{-1}) \ , \ \by_2 \sim\mathcal{CN}(F\bh, \sigma_{n}^2 \bB_2^{-1}) \ .
\eeas
\hfill $\diamond$
\end{lemma}

\begin{proof}
From \er{3observations2}, note the observations ${\bf u} \define (u_1, \ldots , u_T)^T$ are
conditionally independent given $\bh$, with conditional distributions $u_t \sim\mathcal{CN}(\bh^T \bx_t, \sigma_{n}^2)$ for $1 \leq t \leq K$, and $u_t \sim\mathcal{CN}(F\bh^T \bx_t, \sigma_{n}^2)$ otherwise. The conditional distributions of $\by_1$ and $\by_2$ in the lemma follow immediately by substituting these conditional pdfs into \er{ys}.

From the Neyman-Fisher Theorem \cite[pg.~117]{kay}, to prove sufficiency it suffices to show $p( {\bf u} ; F, \sigma_h^2 )$ can be factored into a product $g( \by_1 , \by_2, F, \sigma_h^2)f({\bf u})$, where $f$ does not depend on $\by_1 , \by_2, F$ or $\sigma_h^2$, and $g$ does not depend on $\bu$.

To this end, we can express this pdf in terms of the conditional pdf as
\beas
p( {\bf u} ; F, \sigma_h^2 ) \ = \ E \biggl[ p( {\bf u} | {\bf h} ; F, \sigma_h^2 ) \biggr] \ ,
\eeas
where $E[ \cdot ]$ denotes expectation with respect to $\bh$.
Since $u_1, \ldots , u_T$ are conditionally independent given $\bh$, we can write
\bea
	&& (\pi \sigma_n^2)^T \cdot p\left( \bu ; F, \sigma_h^2 \right) \  = \ E\left[ \exp \left( -\frac{1}{\sigma_n^2}\sum_{t=1}^{K}|u_t - \bh^T\bx_t|^2
-\frac{1}{\sigma_n^2}\sum_{t=K+1}^{T}|u_t - F\bh^T \bx_t|^2 \right) \right] \nn\\
	&=&  E \Biggl[ \exp \left( -\frac{1}{\sigma_n^2}\sum_{t=1}^{K}\left\{ |\bh^T\bx_t|^2-u_t^*\bh^T\bx_t-u_t\bh^H\bx_t^*\right\} \right) \nn\\
	&& \times \exp \left( -\frac{1}{\sigma_n^2}\sum_{t=K+1}^{T}\left\{ |F\bh^T\bx_t|^2-u_t^*F\bh^T\bx_t-u_tF^*\bh^H\bx_t^* \right\} \right) \Biggr] \exp\left(-\frac{|\bu|^2}{\sigma_n^2}\right)  \nn\\
&=&  E \Biggl[ \exp \left( \frac{1}{\sigma_n^2} \left\{ 2\Real[ \bh^H \bB_1 \by_1] +2 \Real[ F^* \bh^H \bB_2 \by_2] -\sum_{t=1}^{K} |\bh^T\bx_t|^2 - \sum_{t=K+1}^{T} |F\bh^T\bx_t|^2  \right\}  \right) \Biggr]  \nn\\
	&& \times \exp\left(-\frac{|\bu|^2}{\sigma_n^2}\right) \nn  \ .
\eea
Equate the
first term with $(\pi \sigma_n^2)^T g( \by_1 , \by_2, F, \sigma_h^2)$, and the second with $f({\bf u})$. Note $f$ does not depend on $\by_1 , \by_2, F,$ or $\sigma_h^2$, while $g$ depends on
$\by_1 , \by_2, F$ and $\sigma_h^2$ (through the expectation), but not $\bu$. This completes the proof.
\end{proof}

Training sequences for MIMO channel estimation are often chosen to be orthogonal and equal-energy, so $\sum_{t=1}^T \bx_t \bx_t^H = (PT/N)\bI_N$. In this section, we assume
the load impedance \er{zL_switch} switches halfway through the training sequence, so $K=T/2$, and the training sequences are also equal-energy and orthogonal over the first and last $K$ symbols, which implies
\beq\label{Xs}
\bB_1 \ = \ \bB_2 \ = \ \left( \frac{PT}{2N} \right) \bI_N \ ,
\eeq
where $\bB_1, \bB_2$ are defined in \er{3ss}. For example, this can be achieved
by using a normalized discrete Fourier transform
(DFT) matrix \cite[eq. 10]{bigu}.

With these assumptions, we
now derive the maximum-likelihood estimate of the parameter vector
\beq\label{eq_theta}
\btt \ \define \ \begin{bmatrix}
	F &
	\sigma_h^2
\end{bmatrix}^T\ ,
\eeq
based on the single-packet sufficient statistics \er{ys}, where
$F$ is defined in \er{3F} and $\sigma_h^2$ in \er{h_pca}. This estimate is defined by
\bea\label{3MLdef}
\btt_{ML} \ \define \ {\rm arg} \max_{\btt} p( \by_1 , \by_2 ; \btt) \ .
\eea
The following theorem shows that these estimators can easily be calculated from the principal component of a sample covariance matrix.

\begin{theorem}[Single-Packet ML Estimators] \label{3ThrmML} Let $\by_1$ and $\by_2$ be the sufficient statistics in \er{ys}, where $F$ and $\sigma_h^2$ are unknown constants. Consider the
sample covariance
\beq\label{samp_cov}
	\bS \ \define \ \frac{1}{N}\begin{bmatrix}
		\by_1^H\by_1& \by_2^H\by_1\\
		\by_1^H\by_2 & \by_2^H\by_2
	\end{bmatrix}\ .
	\eeq
Let $\eta_1$ be the largest eigenvalue of $\bS$, and
$\hat{\be_1}=[E_1,E_2]^T$ any associated unit eigenvector.
Then the maximum-likelihood estimates of $F$ and $\sigma_h^2$ are given by
\beq\label{MLE_iid}
	\hat{\btt}_{ML} \ = \ \begin{bmatrix}
		\hat{F}_{ML}\\
		\hat{\sigma_h^2}
	\end{bmatrix}
	\ = \ \begin{bmatrix}
		{E_2}/{E_1}\\
		|E_1|^2 \max \{ \eta_1-\sigma^2, 0 \}
	\end{bmatrix} \ ,
	\eeq
provided
$E_1\neq 0$, where $\sigma^2 \define \sigma_n^2 (2N/PT)$. For $E_1 =0$ and $\eta_1 > \sigma^2$ the likelihood is maximized in the limit as $F \rightarrow \infty$.  $\hfill\diamond$
\end{theorem}

\begin{proof}  From \er{Xs} and Lemma 2, $\by_1$ and $\by_2$ are conditionally independent given $\bh$, and their conditional distributions are $\by_1 \sim\mathcal{CN}(\bh, \sigma^2 \bI_N)$
and $\by_2 \sim\mathcal{CN}(F\bh, \sigma^2 \bI_N)$,
where $\sigma^2 \define \sigma_n^2 (2N/PT)$. Since $\bh \sim\mathcal{CN}({\bf 0}_N , \sigma_h^2 \bI_N)$ is independent of the noise in \er{ys}, it follows that 
\beq\label{eq_y}
\by \ = \ \begin{bmatrix}
  \by_1^T &
  \by_2^T
\end{bmatrix}^T \ ,
\eeq
is a zero-mean Gaussian random vector with covariance
\bea
	E \left[ \by\by^H \right] &= & \begin{bmatrix}
		( \sigma_h^2 + \sigma^2) \bI_N & \sigma_h^2 F^* \bI_N \\
		\sigma_h^2 F \bI_N & (\sigma_h^2 |F|^2+ \sigma^2)\bI_N
	\end{bmatrix} \ = \ \bC \otimes \bI_N \ ,
	\eea
where $\otimes$ is the Kronecker product \cite{brew} and
\bea\label{cov_iid}
\bC & \define & \bC ( \btt) \ = \ \begin{bmatrix}
		\sigma_h^2 + \sigma^2 & \sigma_h^2 F^*\\
		\sigma_h^2 F & \sigma_h^2 |F|^2+ \sigma^2
	\end{bmatrix} \ .
	\eea
It follows the likelihood function can be written as, 
	\bea\label{pdf_v_iid}
	p\left(\by_1 , \by_2 ;\btt\right) & = & \frac{1}{\det(\pi [\bC \otimes \bI_N])} \exp \left(-\by^H[\bC \otimes \bI_N]^{-1}\by \right)  \ , \nn\\
& = & \det(\pi \bC)^{-N} \exp \left(-\by^H[\bC^{-1} \otimes \bI_N] \by \right)  \ , \nn\\
	&= &\det(\pi \bC)^{-N} \exp \left(-N\Trace\left[\bS\bC^{-1}\right]\right) \ ,
	\eea
	where the last equality is obtained by defining an intermediate matrix $\bY=[\by_1,\by_2]$, noting $\bS=\frac{1}{N}\bY^T\bY^*$ and $\by=\vc \bY$, and an identity $\Trace\left[\bA\bB\bC\bD\right] = \vc^T(\bB)\left[\bC\otimes\bA^T\right]\vc (\bD^T)$ \cite[eq.~2.116]{hjor}.
Note $\bC$ can be written in terms of its eigensystem as
	\bea
	\bC \ = \ \mu_1\be_1\be_1^H +  \mu_2\be_2\be_2^H \ ,
	\eea
	where $\mu_1 \geq \mu_2$ are the ordered eigenvalues and $\be_1, \be_2$ are the associated unit eigenvectors. From \er{cov_iid}, it is easy to verify the following explicit formulas,
	\bea\lb{eigensystem}
	\mu_1 & = & \sigma_h^2(1+|F|^2) + \sigma^2 \ , ~~~\mu_2 \ = \ \sigma^2 \nn\\
	\be_1 & = & \frac{1}{\sqrt{1+|F|^2}}\begin{bmatrix}
		1\\
		F
	\end{bmatrix} \ , ~~~ \be_2 \ = \ \frac{1}{\sqrt{1+|F|^2}}\begin{bmatrix}
		-F^*\\
		1
	\end{bmatrix} \ .
	\eea

To find maximum-likelihood estimates of $F$ and $\sigma_h^2$, we proceed in two steps: First, we find conditions on $\mu_1, \be_1, \be_2$ that achieve the maximum in \er{3MLdef}. (Note the value of $\mu_2$ is fixed at $\sigma^2$.) Second, we use \er{eigensystem} to translate these conditions into values of
$F$ and $\sigma_h^2$.
From \er{cov_iid}, observe
\beq\label{eq-C_inv}
	\bC^{-1} \ = \  \mu_1^{-1} \be_1\be_1^H + \mu_2^{-1} \be_2\be_2^H \ ,
	\eeq
so the log-likelihood function can be expressed in terms of the eigen-system as
\bea\label{eq_llf}
	\ln p\left(\by_1 , \by_2 ;\btt\right) & = & -N\ln \det(\pi \bC) -N\Trace\left[\bS\bC^{-1}\right] \nn\\
	&=&  -N\ln(\pi \mu_1\mu_2)  -\frac{N\be_1^H\bS\be_1}{\mu_1}-\frac{N\be_2^H\bS\be_2}{\mu_2}\nn\\
	&=& -N\ln(\pi \mu_1\mu_2) +  \left(\mu_2^{-1}-\mu_1^{-1}\right)N\be_1^H\bS\be_1-{\mu_2^{-1}}{N\Trace[\bS]}\ ,
	\eea
	where the last equality follows by observing $\bU\define [\be_1,\be_2]$ is unitary, and
	\beq
	\Trace [\bS] \ = \ \Trace[\bU^H\bS\bU] \ = \ \be_1^H\bS\be_1+\be_2^H\bS\be_2 \ .
	\eeq

From \er{eigensystem}, the coefficient $\mu_2^{-1}-\mu_1^{-1}$ in \er{eq_llf} is non-negative. Since $\bS$ is Hermitian and positive semi-definite, it has real eigenvalues, say $\eta_1 \geq \eta_2 \geq 0$. From the Rayleigh-Ritz theorem, $\be_1^H\bS\be_1 \leq \eta_1$ with equality if and only if
$\be_1$ is an eigenvector of $\bS$ corresponding to $\eta_1$. Since $\Trace [\bS] = \eta_1 + \eta_2$, it follows
	\bea\label{llf_max}
	\ln p\left(\by_1 , \by_2 ;\btt\right) & \leq &
	-N\ln(\pi \mu_1\mu_2)  -\frac{N\eta_1}{\mu_1}-\frac{N\eta_2}{\mu_2}\nn\\
	&=& N\left[\ln \left(\frac{\eta_1}{\mu_1}\right)-\frac{\eta_1}{\mu_1}\right] -N\ln(\pi \eta_1\mu_2) -\frac{N\eta_2}{\mu_2} \ ,
	\eea
with equality if and only if $\mu_1 = \mu_2$ or $\be_1$ is an eigenvector of $\bS$ corresponding to $\eta_1$. Note the function $\ln x - x$ is concave and uniquely maximized at $x=1$. It follows that the bracketed term in \er{llf_max} is maximized over $\mu_1 \geq \mu_2 = \sigma^2$ by choosing $\mu_1 = \max \{ \eta_1 , \sigma^2 \}$.

Finally, we translate these conditions into values of $F$ and $\sigma_h^2$: If $\eta_1 \leq \sigma^2$, then $\mu_1 = \mu_2 = \sigma^2$ and the likelihood \er{eq_llf} does not depend on $F$. From \er{eigensystem}, it follows the likelihood is maximized by $\hat{\sigma}_h^2 = 0$ and any value of $F$; In particular, \er{MLE_iid} maximizes the likelihood. However, if $\eta_1 > \sigma^2$, then $\mu_1 = \eta_1$ is optimal and hence $\be_1$ must be an eigenvector $(E_1,E_2)^T$ of $\bS$ corresponding to $\eta_1$. For $E_1 \neq 0$, the unique solution of the equations $\sigma_h^2 = \eta_1$ and $\be_1 = (E_1,E_2)^T$ in \er{eigensystem} is given by \er{MLE_iid}. For $E_1=0$, no finite $F$ solves these equations; rather, the solution is approached in the limit as $F \rightarrow \infty$.
\end{proof}

The theorem above shows the maximum-likelihood estimators can be expressed in terms of eigensystem of the sample covariance $\bS$. Since $\bS$ is a $2 \times 2$ Hermitian matrix, it is possible to derive closed-form formulas for the eigensystem. This leads
to closed-form expressions for the estimators, which are given in the following corollary.

\begin{corollary}[Closed-Form ML Estimators]\label{3CFML} 
Denote the elements of the sample covariance $\bS$ in \eqref{samp_cov} by
\footnote{Note the entry $S_{ij}=\by_j^H\by_i/N$ corresponds to $P_{ji}$ in our prequel \cite[eq.~31]{wu2}, where $1\leq i, j\leq 2$. }
	\beq\label{Ssymbolic}
	\bS \ = \ \frac{1}{N}\begin{bmatrix}
		\by_1^H\by_1& \by_2^H\by_1\\
		\by_1^H\by_2 & \by_2^H\by_2
	\end{bmatrix} \ \define  \ \begin{bmatrix}
		S_{11}& S_{12}\\
		S_{21} & S_{22}
	\end{bmatrix}\ .
	\eeq
The maximum-likelihood estimators in \er{MLE_iid} can be expressed in closed-form as
	\bea\label{FMLiid}
	\hat{F}_{ML} & = & \frac{S_{22}-S_{11}+\sqrt{(S_{22}-S_{11})^2+4|S_{12}|^2}}{2S_{12}} \ , \\
\hat{\sigma}_{h}^2 & = & \frac{|\hat{F}_{ML}|^2}{|\hat{F}_{ML}|^2+1} \max \left\{ S_{11}+S_{12}\hat{F}_{ML} -\sigma^2 , 0 \right\} \label{sigmahiid} \ , 
	\eea
	provided $S_{12}\neq 0$. $\hfill\diamond$
\end{corollary}

\begin{proof} 
Since $\bS$ is Hermitian, $S_{11}$ and $S_{22}$ are real and non-negative and $S_{12}=S_{21}^*$. The eigenvalues are given by the two non-negative roots of the polynomial
\[
	{\rm det}[ \bS - \eta \bI] \ = \ \eta^2 - (S_{11}+S_{22})\eta+ S_{11}S_{22}-|S_{12}|^2 \ ,
\]
of which the larger is
\beq\label{3eta1}
\eta_1 \ = \frac{S_{22}+S_{11} + \sqrt{(S_{11}-S_{22})^2+4|S_{12}|^2}}{2} \ .
\eeq

Next we find a unit eigenvector $\be_1 = [ E_1, E_2]^T$ associated with $\eta_1$. Any such vector must satisfy $|E_1|^2+|E_2|^2=1$ and
 \bes
\begin{bmatrix}
		S_{11}-\eta_1 & S_{12} \\
		S_{21} & S_{22}-\eta_1
	\end{bmatrix} \begin{bmatrix}
		E_1 \\
		E_2
	\end{bmatrix} \ = \ \begin{bmatrix}
		0 \\
        0
	\end{bmatrix} \ .
\ees
It is easy to verify that one solution of these equations is
 \beq\label{3e1}
 \begin{bmatrix}
		E_1 \\
		E_2
	\end{bmatrix} \ = \ \frac{1}{\sqrt{| \eta_1 - S_{11}|^2 + |S_{12}|^2}} \begin{bmatrix}
		S_{12} \\
		\eta_1 - S_{11}
	\end{bmatrix} \ .
\eeq
From \er{MLE_iid}, it follows
\[
\hat{F}_{ML} \ = \ \frac{\eta_1 - S_{11}}{S_{12}} \ , \ \hat{\sigma}_h^2 \ = \ \frac{| \eta_1 - S_{11}|^2}{| \eta_1 - S_{11}|^2 + |S_{12}|^2} \max \left\{ \eta_1 -\sigma^2 , 0 \right\} \ ,
\]
Substituting \er{3eta1} into the first equation gives \er{FMLiid}; substituting $\eta_1 = S_{11}+S_{12}\hat{F}_{ML}$ into the second gives \er{sigmahiid}.
\end{proof}

We note \er{FMLiid} is similar to estimators that arise in the errors-in-variables regression literature, cf. \cite[pg.~294, Case 4]{gill}.

Finally, in order to evaluate the efficiency of these estimators, we now derive the Cram\'{e}r-Rao Bound for the error covariance of these estimators,
\[
{\bf C}_{\hat{\boldsymbol \theta}} \ \triangleq \ E_{\by_1 , \by_2 ;\btt } \left[ \left(\hat{\boldsymbol \theta}-{\boldsymbol \theta} \right)\left(\hat{\boldsymbol \theta}-{\boldsymbol \theta} \right)^H \right] \ ,
\]
where $E_{\by_1 , \by_2 ;\btt }$ denotes expectation with respect to the pdf $p\left(\by_1 , \by_2 ;\btt\right)$ in \er{llf_max}.  The Cram\'{e}r-Rao Bound (CRB) holds\footnote{The CRB holds because the support of likelihood function (LF) \eqref{pdf_v_iid}  does not depend on $\btt$, and the first two derivatives of the LF w.r.t. $\btt$ exist and has absolute integrability w.r.t. $\by_1$ and $\by_2$ \cite[eq.~1]{lu}. }  and is given by
\[
{\bf C}_{\hat{\boldsymbol \theta}} 
\ \geq \ \boldsymbol{\mathcal{I}}({\boldsymbol \theta})^{-1} \ ,
\]
where $\hat{\boldsymbol \theta}$ is any unbiased estimator of ${\boldsymbol \theta} = [ \theta_1 , \theta_2 ]^T = [ {F}, {\sigma}_h^2 ]^T$, and $\boldsymbol{\mathcal{I}}({\boldsymbol \theta})$ is the Fisher information matrix (FIM) \cite[pg.~529]{kay}
\beq\label{FIMij_ch3}
[\boldsymbol{\mathcal{I}}(\btt)]_{ij} \ = \ N\Trace\left[\bC^{-1}\frac{\partial \bC}{\partial \theta_i^*}\bC^{-1}\frac{\partial \bC}{\partial \theta_j}\right] \ ,
\eeq
where $\bC$ is given in \er{cov_iid}. Here we use the approach described in \cite[Sec.~15.7]{kay} to state the FIM in an equivalent form convenient for complex $\btt$.

\ifcomments
\subsection{Detail: Derivation of FIM}
To simplify matters, first express $\bC$ in terms of $\btt$ as
\bea\label{cdef}
\bC & \define & \bC ( \btt) \ = \ \begin{bmatrix}
		\theta_2 + \sigma^2 & \theta_2 \theta_1^*\\
		\theta_2 \theta_1 & \theta_2 |\theta_1|^2+ \sigma^2
	\end{bmatrix} \ .
	\eea
so the FIM entries are
\beas
[\boldsymbol{\mathcal{I}}(\btt)]_{11} & = &  N\Trace\left[\bC^{-1}\frac{\partial \bC}{\partial \theta_1^*}\bC^{-1}\frac{\partial \bC}{\partial \theta_1}\right] \\
& = & \frac{N}{( {\rm det} \bC )^2} \Trace \Biggl\{
\begin{bmatrix}
		\theta_2 |\theta_1|^2+ \sigma^2 & -\theta_2 \theta_1^*\\
		-\theta_2 \theta_1 & \theta_2+ \sigma^2
	\end{bmatrix}
\begin{bmatrix}
		0 & \theta_2 \\
		0 & \theta_2 \theta_1
	\end{bmatrix}
\begin{bmatrix}
		\theta_2 |\theta_1|^2+ \sigma^2 & -\theta_2 \theta_1^*\\
		-\theta_2 \theta_1 & \theta_2+ \sigma^2
	\end{bmatrix}
\begin{bmatrix}
		0 & 0\\
		\theta_2 & \theta_2\theta_1^*
	\end{bmatrix}
\Biggr\} \\
& = & \frac{N}{( {\rm det} \bC )^2} \Trace \Biggl\{
\begin{bmatrix}
		0 & \theta_2 \sigma^2 \\
		0 & \theta_2 \theta_1 \sigma^2
	\end{bmatrix}
\begin{bmatrix}
		-\theta_2^2 \theta_1^* & -\theta_2^2(\theta_1^*)^2\\
		(\theta_2+ \sigma^2)\theta_2 & (\theta_2+ \sigma^2)\theta_2\theta_1^*
	\end{bmatrix}
\Biggr\} \\
& = & \frac{N}{( {\rm det} \bC )^2} \Biggl\{ (\theta_2+ \sigma^2)\theta_2^2 \sigma^2 +
(\theta_2+ \sigma^2)\theta_2^2 | \theta_1 |^2 \sigma^2
\Biggr\} \\
& = & \frac{N(\theta_2+ \sigma^2)\theta_2^2 \sigma^2 (1+| \theta_1 |^2)}{\left( \sigma^2 \theta_2 (1+ | \theta_1 |^2) + \sigma^4 \right)^2} \ = \
\frac{N(1+| \theta_1 |^2)\theta_2^2(\theta_2/\sigma^2+1) }{\left( \theta_2 (1+ | \theta_1 |^2) + \sigma^2 \right)^2} \ = \ \frac{N(1+| \theta_1 |^2)\rho^2(\rho+1) }{\left( \rho (1+ | \theta_1 |^2) + 1 \right)^2}
\eeas
where $\rho \define \theta_2/\sigma^2$.
\beas
[\boldsymbol{\mathcal{I}}(\btt)]_{12} & = &  N\Trace\left[\bC^{-1}\frac{\partial \bC}{\partial \theta_1^*}\bC^{-1}\frac{\partial \bC}{\partial \theta_2}\right] \\
& = & \frac{N}{( {\rm det} \bC )^2} \Trace \Biggl\{
\begin{bmatrix}
		\theta_2 |\theta_1|^2+ \sigma^2 & -\theta_2 \theta_1^*\\
		-\theta_2 \theta_1 & \theta_2+ \sigma^2
	\end{bmatrix}
\begin{bmatrix}
		0 & \theta_2 \\
		0 & \theta_2 \theta_1
	\end{bmatrix}
\begin{bmatrix}
		\theta_2 |\theta_1|^2+ \sigma^2 & -\theta_2 \theta_1^*\\
		-\theta_2 \theta_1 & \theta_2+ \sigma^2
	\end{bmatrix}
\begin{bmatrix}
		1 & \theta_1^*\\
		\theta_1 & |\theta_1|^2
	\end{bmatrix}
\Biggr\} \\
& = & \frac{N}{( {\rm det} \bC )^2} \Trace \Biggl\{
\begin{bmatrix}
		0 & \theta_2 \sigma^2 \\
		0 & \theta_2 \theta_1 \sigma^2
	\end{bmatrix}
\begin{bmatrix}
		\sigma^2 & \sigma^2\theta_1^*\\
		\sigma^2 \theta_1 & \sigma^2 |\theta_1|^2
	\end{bmatrix}
\Biggr\} \\
& = & \frac{N}{( {\rm det} \bC )^2} \Biggl\{ \sigma^4 \theta_1 \theta_2 +
\sigma^4 \theta_1 \theta_2 | \theta_1 |^2
\Biggr\} \\
& = & \frac{N\sigma^4 \theta_1 \theta_2 (1+| \theta_1 |^2)}{\left( \sigma^2 \theta_2 (1+ | \theta_1 |^2) + \sigma^4 \right)^2} \ = \
\frac{N\theta_1 \theta_2 (1+| \theta_1 |^2) }{\left( \theta_2 (1+ | \theta_1 |^2) + \sigma^2 \right)^2} \ = \ \frac{N\theta_1 \rho (1+| \theta_1 |^2) }{\sigma^2\left( \rho (1+ | \theta_1 |^2) + 1 \right)^2}
\eeas
\fi

It follows
\beq\label{FIM_sp_ch3}
\mathcal{I}(\btt) \ = \ \frac{N(1+|F|^2)}{\left[\sigma_h^2(1+|F|^2) + \sigma^2 \right]^2}\begin{bmatrix}
	\sigma_h^4 \left(1+ \sigma_h^2/\sigma^2 \right)& F\sigma_h^2 \\
	F^*\sigma_h^2& 1+|F|^2
\end{bmatrix} \ ;
\eeq
hence the CRB is given by
\beq
\boldsymbol{\mathcal{I}}({\boldsymbol \theta})^{-1} \ = \ \frac{\sigma_h^2 (1+|F|^2)+\sigma^2}{N \sigma_h^4 (1+|F|^2)} \begin{bmatrix}
	\sigma^2(1+|F|^2) & -F\sigma^2\sigma_h^2\\
	-F^*\sigma^2\sigma_h^2& \sigma_h^4 ( \sigma_h^2 + \sigma^2 )
\end{bmatrix}  \ .
\eeq

We are interested primarily in bounds on the mean-squared error of unbiased estimators
of
$F$ and  $\sigma_h^2$. These are given, respectively, by the diagonal entries of \er{CRB_sp_ch3}:
\beq\label{CRB_sp_ch3}
\mathcal{C}_F(\btt) \ \define \ \frac{\sigma^2 \sigma_h^2 (1+|F|^2) + \sigma^4}{N \sigma_h^4} \ , ~~~~
\mathcal{C}_{\sigma_h^2}(\btt) \ \define \ \frac{\left( \sigma_h^2 (1+|F|^2)+\sigma^2 \right)\left( \sigma_h^2 +\sigma^2 \right)}{N(1+|F|^2)} \ .
\eeq

\section{Estimators for Multiple Packets}\lb{3secIV}
In the last section, we derived estimators for $F$ and $\sigma_h^2$ based on a single training packet. In this section, we consider estimators based on multiple packets where the channel evolves in an unknown way from packet to packet.

Suppose the transmitter sends a sequence of $L$ identical training packets to the receiver. During reception of each packet, the receiver load shifts in the same
way as described in \er{zL_switch}. We assume the channel
is constant {\em within} a packet, but varies from packet to packet in a random way.
Under these assumptions, the signal observed during the $k$-th packet can be described
by a model similar to \er{3observations2}:
\beq\label{3observationsMP}
u_{k,t}  \ = \ \begin{cases}
	\bh_k^T \bx_t + n_{k,t} \ , & 1 \leq t \leq K \ , \\
	F \bh_k^T \bx_t  + n_{k,t} \ , & K < t \leq T \ ,
\end{cases}
\eeq
where $F$ is still defined by \er{3F}, $\bh_k$ is the channel during the $k$-th packet, and the noise variable $n_{k,t} \sim\mathcal{CN}(0, \sigma_n^2)$ are i.i.d.
We can express these observations in a compact matrix form as
\beq\label{3observationsMP2}
\bU_1 = \bH \bX_1 + \bN_1 \ , \bU_2 = F \bH \bX_2 + \bN_2
\eeq
where $\bX_1 \define [ \bx_1 , \ldots, \bx_K] \in \mbbC^{N\times K}$, $\bX_2 \define [ \bx_{K+1} , \ldots, \bx_T]\in \mbbC^{N\times (T-K)}$,
\beq\label{H_ch3}
\bH \ \define \ \begin{bmatrix}
  \bh_1^T\\
  \vdots\\
  \bh_L^T
\end{bmatrix} \ = \ [ \bh_1 , \ldots, \bh_L]^T \in \mbbC^{L\times N} \ ,
\eeq and $\bU_1 \in \mbbC^{L\times K}, \bU_2 \in \mbbC^{L\times (N-K)}, \bN_1 \in \mbbC^{L\times K}$ and $\bN_2 \in \mbbC^{L\times (N-K)}$ are defined analogously. It follows $\bN_1$ and $\bN_2$ are independent random matrices with i.i.d. $\mathcal{CN}(0, \sigma_n^2)$ entries. Note the horizontal dimension of $\bH$ represents space, while vertical dimension time.  Here $\bH$ models Rayleigh fading path gains which are uncorrelated in space but not necessarily in time. This implies the columns of $\bH$ are i.i.d. zero-mean Gaussian random vectors with an temporal correlation matrix $\sigma_h^2 \bC_\bH \in \mbbC^{L\times L}$. Here we assume the correlation structure of $\bH$ is known except for the power, so $\bC_\bH$ is known but $\sigma_h^2$ is unknown.  As in the last section, we assume $\bX_1$ and $\bX_2$ are known
at the receiver, $K=T/2$,
and the training sequences are equal-energy and orthogonal over the first and last $K$ symbols, which implies
\beq\label{XsMP}
\bX_1\bX_1^H \ = \ \bX_2\bX_2^H \ = \ \left( \frac{PT}{2N} \right) \bI_N \ .
\eeq

The goal of this paper is to derive estimators for $\bH, F$ and $\sigma_h^2$ based on the observations \er{3observationsMP2}. As in the last section, we approach the problem in two steps: In this section, we consider joint maximum-likelihood estimation of $F$ and $\sigma_h^2$, treating $\bH$ as a nuisance parameter. In Sec.~\ref{3secV}, we will explore estimators for $\bH$ given $F$ and $\sigma_h^2$ through numerical examples. The following lemma generalizes Lemma 2 to multiple packets.

\begin{lemma}[Multi-Packet Sufficient Statistics]\lb{sslemmaMP}  Consider the observations $\bU_1, \bU_2$ defined in
\eqref{3observationsMP2}, where $\bX_1,\bX_2$ are known training sequences and $F$ and $\sigma_h^2$ are unknown constants. Then
\begin{eqnarray}\label{ysMP}
\bY_1 \ \define \ \left( \frac{2N}{PT} \right) \bU_1 \bX_1^H \ , \ \ \bY_2 \ \define \ \left( \frac{2N}{PT} \right) \bU_1 \bX_2^H  ,
\end{eqnarray}
are sufficient statistics to estimate $F$ and $\bC_\bH$ based on the observation $\bU_1, \bU_2$. Moreover, $\bY_1-\bH$ and $\bY_2-F\bH$ are independent random matrices with i.i.d. $\mathcal{CN}(0, \sigma^2)$ entries,
where $\sigma^2 \define 2N\sigma_n^2/P$. \hfill $\diamond$
\end{lemma}

\begin{proof}
From \er{3observationsMP2} and \er{XsMP}, we have $\bY_1 = \bH + (2N/PT)\bN_1 \bX_1^H$. Note the rows
of $(2N/PT)\bN_1 \bX_1^H$ are i.i.d. with covariance $\sigma_n^2 (2N/PT)^2 \bX_1 \bX_1^H = \sigma^2 \bI_N$, where the last step follows from \er{XsMP}. Similarly, $\bY_2 = F\bH + (2N/PT) \bN_2 \bX_2^H$, where the last matrix has i.i.d $\mathcal{CN}(0, \sigma^2)$ entries. It follows  $\bY_1-\bH$ and $\bY_2-F\bH$ are independent random matrices with i.i.d. $\mathcal{CN}(0, \sigma^2)$ entries.

From the Neyman-Fisher Theorem \cite[pg.~117]{kay}, to prove sufficiency it suffices to show $p( \bU_1 , \bU_2 ; F, \sigma_h^2 )$ can be factored into a product $g( \bY_1 , \bY_2, F, \sigma_h^2 )f(\bU_1, \bU_2 )$, where $f$ does not depend on $\bY_1 , \bY_2, F$ or $\sigma_h^2$, and $g$ does not depend on $\bU_1, \bU_2$.

To this end, we can express this pdf in terms of the conditional pdf as
\beas
p( \bU_1 , \bU_2 ; F, \sigma_h^2  ) \ = \ E_\bH \biggl[ p( \bU_1 , \bU_2 | \bH ; F, \sigma_h^2 ) \biggr] \ ,
\eeas
where $E_\bH [ \cdot ]$ denotes expectation with respect to $\bH$.
Since $\bU_1 , \bU_2$ are conditionally independent given $\bH$, we can write
\bea
	&& (\pi \sigma_n^2)^{LT} \cdot p\left( \bU_1 , \bU_2 ; F, \sigma_h^2 \right) \nn\\ 
	 &=& E_\bH\left[ \exp \left( -\frac{1}{\sigma_n^2} \parallel \bU_1 - \bH \bX_1 \parallel_F^2
-\frac{1}{\sigma_n^2}\parallel \bU_2 - F \bH \bX_2 \parallel_F^2 \right) \right] \nn\\
	&=&  E_\bH \Biggl[ \exp \Biggl( -\frac{1}{\sigma_n^2}\biggl\{ -{\rm Tr}[ \bU_1^H \bH \bX_1 ]
-{\rm Tr}[ (\bH \bX_1)^H \bU_1 ]+{\rm Tr}[ \bX_1^H \bH^H \bH \bX_1 ] - {\rm Tr}[ \bU_2^H F\bH \bX_1 ] \nn\\
	&&
-{\rm Tr}[ (F\bH \bX_2)^H \bU_2 ]+{\rm Tr}[ |F|^2 \bX_2^H \bH^H \bH \bX_2 ] \biggr\} \Biggr) \Biggr] \exp\left(-\frac{1}{\sigma_n^2} \left\{ \parallel \bU_1 \parallel_F^2+\parallel \bU_2 \parallel_F^2 \right\}\right)  \nn\\
&=&  E_\bH \Biggl[ \exp \Biggl( \frac{1}{\sigma_n^2}\biggl\{ 2 {\rm ReTr}[ \bH^H \bU_1 \bX_1^H ] + 2 {\rm ReTr}[ F^* \bH^H \bU_2 \bX_2^H ] -{\rm Tr}[  \bH^H \bH \bX_1 \bX_1^H ] \nn\\
	&&
-|F|^2{\rm Tr}[ \bH^H \bH \bX_2 \bX_2^H ] \biggr\} \Biggr) \Biggr] \exp\left(-\frac{1}{\sigma_n^2} \left\{ \parallel \bU_1 \parallel_F^2+\parallel \bU_2 \parallel_F^2 \right\}\right)  \nn\\
&=&  E_\bH \Biggl[ \exp \Biggl( \frac{1}{\sigma^2}\biggl\{ 2 {\rm ReTr}[ \bH^H \bY_1 ] + 2 {\rm ReTr}[ F^* \bH^H \bY_2 ] - (1 + |F|^2) \parallel \bH \parallel_F^2 \biggr\} \Biggr) \Biggr] \nn \\ && \times \exp\left(-\frac{1}{\sigma_n^2} \left\{ \parallel \bU_1 \parallel_F^2+\parallel \bU_2 \parallel_F^2 \right\}\right)   \ , \label{3LemmaMP}
\eea
where $\parallel \bA \parallel_F^2 = {\rm Tr}[ \bA^H \bA ]$ denotes the Frobenius norm.
Here the third equality follows from the identities $2{\rm ReTr}[\bA]={\rm Tr}[\bA]+{\rm Tr}[\bA^H]$ and ${\rm Tr}[\bA \bB] = {\rm Tr}[\bB \bA]$, and the fourth equality follows from \er{XsMP} and the definition of $\sigma^2$. In \er{3LemmaMP}, denote the first factor
by $(\pi \sigma_n^2)^T g( \bY_1 , \bY_2, F, \sigma_h^2)$, and the second by $f(\bU_1, \bU_2 )$. Note $f$ does not depend on $\bY_1 , \bY_2, F,$ or $\sigma_h^2$, while $g$ depends on
$\bY_1 , \bY_2, F$ and $\sigma_h^2$ (through the expectation), but not $\bU_1, \bU_2$. This completes the proof.
\end{proof}

We now present maximum-likelihood estimators for the parameter vector $\btt$ defined in \er{eq_theta}
using sufficient statistics \er{ysMP}. This estimate is defined by
\bea\label{3MLdefMP}
\hat{\btt}_{ML} \ \define \ {\rm arg} \max_{\btt} p( \bY_1 , \bY_2 ; \btt) \ .
\eea
The following theorem shows that these estimators can be calculated from a scalar optimization.

\begin{theorem}[Multiple-Packet ML Estimators]\label{3ThrmMLMP}
  Let $\bY_1$ and $\bY_2$ be the sufficient statistics in \er{ysMP}, where $F$ and $\sigma_h^2$ are unknown constants. Consider the matrix
\beq\label{samp_covMP}
	\bS ( \mu ) \ \define \ \begin{bmatrix}
		S_{11} ( \mu) & S_{12} ( \mu) \\
		S_{21} ( \mu) & S_{22} ( \mu)
	\end{bmatrix}\ .
	\eeq
where $S_{ij} ( \mu ) \ \define \ \frac{1}{N}{\rm Tr}  \left[ \mu \bC_\bH \left(\mu \bC_\bH + \sigma^2 \bI_L \right)^{-1} \bY_i \bY_j^H \right]$, $1 \leq i,j \leq 2$.
Define
\bea\label{3muhatMP}
\hat{\mu} \ \define
{\rm arg} \max_{\mu \geq 0} \ \left[ \eta ( \mu ) - \sigma^2 \ln {\rm det} [\mu \bC_\bH + \sigma^2 \bI_L ] \right] \ ,
\eea
where $\eta ( \mu )$ is the largest eigenvalue of $\bS ( \mu)$:
\beq\label{3etaMP}
\eta ( \mu ) \ \define \frac{S_{22}(\mu)+S_{11}(\mu) + \sqrt{(S_{11}(\mu)-S_{22}(\mu) )^2+4|S_{12}(\mu)|^2}}{2} \ .
\eeq

Let $\hat{\be_1}=[E_1,E_2]^T$ be any unit eigenvector of $\bS ( \hat{\mu})$ corresponding to the eigenvalue $\eta ( \hat{\mu} )$.
Then the maximum-likelihood estimates of $F$ and $\sigma_h^2$ are given by
\beq\label{3MLEMP}
	\hat{\btt}_{ML} \ = \ \begin{bmatrix}
		\hat{F}_{ML}\\
		\hat{\sigma_h^2}
	\end{bmatrix}
	\ = \ \begin{bmatrix}
		{E_2}/{E_1}\\
		|E_1|^2 \hat{\mu}
	\end{bmatrix} \ ,
	\eeq
provided
$E_1\neq 0$, where $\sigma^2 \define \sigma_n^2 (2N/PT)$. For $E_1 =0$ and $\hat{\mu} > 0$ the likelihood is maximized in the limit as $F \rightarrow \infty$.  $\hfill\diamond$
\end{theorem}

\begin{proof}  For any matrix $\bA$, denote the $kj$-th element and $k$-th row by $[\bA]_{kj}$ and $[\bA]_{k}$, respectively. Let $\bC_\bH = \bV^H {\rm diag}[ \lambda_1 , \ldots , \lambda_L] \bV$ be an eigendecomposition of $\bC_\bH$, where $\lambda_1 \geq \ldots \geq \lambda_L \geq 0$ are eigenvalues of $\bC_\bH$, and $\bV$ is a unitary matrix such that $\bV \bV^H = \bV^H \bV = \bI_L$.
It follows the elements of $\bV \bH$  are independent with $[\bV \bH]_{kj} \sim\mathcal{CN}(0, \sigma_h^2 \lambda_k)$.

For $1 \leq k \leq L$, let $\bw_{k1} \define [ \bV \bY_1]_k$ and $\bw_{k2} \define [ \bV \bY_2]_k$. From Lemma~\ref{sslemmaMP}, $\bw_{k1}$ and $\bw_{k2}$ are conditionally independent given $[\bV\bH]_k$, with conditional distributions
$\bw_{k1} \sim\mathcal{CN}([\bV\bH]_k, \sigma^2 \bI_N)$
and $\bw_{k2} \sim\mathcal{CN}(F[\bV\bH]_k, \sigma^2 \bI_L)$. Since  $[\bV\bH]_k \sim\mathcal{CN}({\bf 0}_N , \sigma_h^2 \lambda_k \bI_N )$ is independent of the noise in \er{ysMP}, it follows $\bw_k \define ( \bw_{k1}, \bw_{k2})^T\in\mbbC^{2N}$ is a zero-mean Gaussian random vector\footnote{Note $\bw_k$'s are column vectors, but $\bw_{k1}$'s and $\bw_{k2}$'s are row vectors. } with covariance
\bea
	\bC_{{\bw}_k} \ \define \ E \left[ \bw_k \bw_k^H \right] &= & \begin{bmatrix}
		\left( \sigma_h^2 \lambda_k + \sigma^2 \right) \bI_N & \sigma_h^2 F^* \lambda_k \bI_N \\
		\sigma_h^2 F \lambda_k \bI_N & \left( \sigma_h^2 |F|^2 \lambda_k + \sigma^2 \right) \bI_N
	\end{bmatrix}  =  \bC_k \otimes \bI_N  \ ,
	\eea
where
\bea\label{3Ckdef}
	\bC_k \ \define \  \begin{bmatrix}
		\sigma_h^2 \lambda_k + \sigma^2 & \sigma_h^2 F^* \lambda_k \\
		\sigma_h^2 F \lambda_k & \sigma_h^2 |F|^2 \lambda_k + \sigma^2
	\end{bmatrix}  \ .
	\eea

As in the derivation of \er{eigensystem} in the proof of Theorem~\ref{3ThrmML}, note
$\bC_k$ can be written in terms of its eigensystem as
	\bea
	\bC_k \ = \ \mu_{k1}\be_1\be_1^H +  \mu_2\be_2\be_2^H \ ,
	\eea
	where $\mu_{k1} \geq \mu_2$ are the ordered eigenvalues and $\be_1, \be_2$ are the associated unit eigenvectors. From \er{3Ckdef}, it is easy to verify the following explicit formulas,
	\bea\lb{eigensystemMP}
	\mu_{k1} & = & \mu \lambda_k + \sigma^2 \ , ~~~\mu_2 \ = \ \sigma^2 \nn\\
	\be_1 & = & \frac{1}{\sqrt{1+|F|^2}}\begin{bmatrix}
		1\\
		F
	\end{bmatrix} \ , ~~~ \be_2 \ = \ \frac{1}{\sqrt{1+|F|^2}}\begin{bmatrix}
		-F^*\\
		1
	\end{bmatrix} \ .
	\eea
where $\mu \define \sigma_h^2 (1+|F|^2)$. Note only $\mu_{k1}$ depends on $k$. As in the derivation of \er{eq_llf} in the proof of Theorem~\ref{3ThrmML}, we have
\bea\label{llfMP}
	\ln p\left(\bw_{k1} , \bw_{k2} ;\btt\right) 
	&=& -N\ln(\pi \mu_{k1}\mu_2) +  \left(\mu_2^{-1}-\mu_{k1}^{-1}\right)N\be_1^H\bS_k\be_1
-{\mu_2^{-1}}{N\Trace[\bS_k]} \nn \\
&=& -N\ln(\pi \sigma^2 [ \mu \lambda_k + \sigma^2])
+  \frac{N \mu \lambda_k}{\sigma^2 [ \mu \lambda_k + \sigma^2]} \be_1^H\bS_k\be_1-\sigma^{-2}{N\Trace[\bS_k]}\ \nn \\
&=& B_k + \frac{N}{\sigma^2} \left[ \frac{\mu \lambda_k}{\mu \lambda_k + \sigma^2} \be_1^H\bS_k\be_1 - \sigma^2 \ln(\mu \lambda_k + \sigma^2) \right]
	\eea
where $B_k$ does not depend on $\mu$ or $\be_1$ and
\beq\label{samp_covMPk}
	\bS_k \ \define \ \frac{1}{N}\begin{bmatrix}
		\bw_{k1}\bw_{k1}^H& \bw_{k1}\bw_{k2}^H\\
		\bw_{k2}\bw_{k1}^H & \bw_{k2}\bw_{k2}^H
	\end{bmatrix}\ = \ \frac{1}{N}\begin{bmatrix}
		[ \bV \bY_1 \bY_1^H \bV^H]_{kk} & [ \bV \bY_1 \bY_2^H \bV^H]_{kk}\\
		[ \bV \bY_2 \bY_1^H \bV^H]_{kk} & [ \bV \bY_2 \bY_2^H \bV^H]_{kk}
	\end{bmatrix} \ .
	\eeq
Since $\bw_1, \ldots , \bw_L$ are independent, the
joint probability of $\bY_1$ and $\bY_2$ is then given by
\bea\label{3llrMP}
\ln p( \bY_1 , \bY_2 ; \btt) & = & \sum_{k=1}^L \ln p\left(\bw_{k1} , \bw_{k2} ;\btt\right) \nn \\
&=& B + \frac{N}{\sigma^2} \sum_{k=1}^L \left[ \frac{\mu \lambda_k}{\mu \lambda_k + \sigma^2} \be_1^H\bS_k\be_1 -\sigma^2\ln(\mu \lambda_k + \sigma^2) \right] \nn \\
&=& B + \frac{N}{\sigma^2} \left[ \be_1^H\bS( \mu ) \be_1 -\sigma^2\sum_{k=1}^L \ln(\mu \lambda_k + \sigma^2) \right] \ ,
	\eea
where $B$ does not depend on the parameters and
\beq
	\bS ( \mu ) \ \define \ \sum_{k=1}^L \frac{\mu \lambda_k}{\mu \lambda_k + \sigma^2}  \bS_k 
	\eeq
is the matrix in \er{samp_covMP}. To see this, let
$\Lambda \define {\rm diag}( \lambda_1 , \ldots, \lambda_L)$ and observe
\bea\label{3Sijdef2}
	N\cdot [\bS ( \mu )]_{ij} & = & \sum_{k=1}^L \frac{\mu \lambda_k}{\mu \lambda_k + \sigma^2} [ \bV \bY_i \bY_j^H \bV^H]_{kk} \nn \\
& = & \sum_{k=1}^L  \left[ \mu \Lambda \left(\mu \Lambda + \sigma^2 \bI_L \right)^{-1} \bV \bY_i \bY_j^H \bV^H \right]_{kk} \nn \\
& = & \sum_{k=1}^L  \left[ \mu \bC_\bH \left(\mu \bC_\bH + \sigma^2 \bI_L \right)^{-1} \bY_i \bY_j^H \right]_{kk} \nn \\
& = & {\rm Tr}  \left[ \mu \bC_\bH \left(\mu \bC_\bH + \sigma^2 \bI_L \right)^{-1} \bY_i \bY_j^H \right] \ .
	\eea

To find maximum-likelihood estimates of $F$ and $\sigma_h^2$, we proceed in two steps: First, we find conditions on $\mu$ and $\be_1$ that achieve the maximum in \er{3llrMP}. Second, we use \er{eigensystemMP} to translate these conditions into values of
$F$ and $\sigma_h^2$.

For each $\mu$, the maximum of \er{3llrMP} over $\be_1$ is clearly a unit eigenvector corresponding to the largest eigenvalue of $\bS ( \mu )$. As in the proof of \er{3eta1},
it is easily shown this eigenvalue is given by $\eta ( \mu )$, defined in \er{3etaMP}.
It follows that the maximum-likelihood estimate of $\mu$ is
\beas
\hat{\mu} & \define & {\rm arg} \max_{\mu \geq 0} \left[ \eta ( \mu ) - \sigma^2 \sum_{k=1}^L \ln(\mu \lambda_k + \sigma^2) \right]  \ ,
\eeas
which equals \er{3muhatMP}, since $\sum_{k=1}^L \ln(\mu \lambda_k + \sigma^2) \ = \ \ln {\rm det} [\mu \bC_\bH + \sigma^2 \bI_L ]$.

Finally, we translate these conditions into values of $F$ and $\sigma_h^2$: If $\hat{\mu}=0$, $\bS( \hat{\mu} )$ vanishes and $\ln p( \bY_1 , \bY_2 ; \btt)$ does not depend on $F$. From \er{eigensystemMP}, it follows the likelihood is maximized by $\hat{\sigma}_h^2 = 0$ and any value of $F$; In particular, \er{3MLEMP} maximizes the likelihood. However, if $\hat{\mu} > 0$, then $\bS( \hat{\mu} )$ is not zero and $\be_1$ must be an eigenvector of $\bS( \hat{\mu} )$ corresponding to $\eta ( \hat{\mu} )>0$. For $E_1 \neq 0$, the unique solution of the equations $\hat{\mu} =\sigma_h^2 (1+|F|^2)$ and $\be_1 = (E_1,E_2)^T$ in \er{eigensystemMP} is given by \er{3MLEMP}. For $E_1=0$, no finite $F$ solves these equations; rather, the solution (and maximum) is approached in the limit as $F \rightarrow \infty$. 
\end{proof}

The theorem above reduces the problem of calculating the multi-packet estimators to the problem of solving the scalar optimization \er{3muhatMP}. In general, it appears that
this optimization must be performed numerically. However, we now show this optimization admits a simple, closed form solution in several scenarios of practical interest.

First consider the case of fast-fading, where $\bC_\bH = \bI_L$. In this case, the estimators can be expressed in a simple form similar to those in Corollary~\ref{3CFML}.

\begin{corollary}[ML Estimators for Fast-Fading]\label{3CFMLMP} Let $\bY_1$ and $\bY_2$ be the sufficient statistics in \er{ysMP}, where $F$ and $\sigma_h^2$ are unknown constants. Consider the matrix
\beq\label{3Tdef}
	\bT \ \define \ \frac{1}{N} \begin{bmatrix}
		{\rm Tr}  \left[ \bY_1 \bY_1^H \right] & {\rm Tr}  \left[ \bY_1 \bY_2^H \right] \\
		{\rm Tr}  \left[ \bY_2 \bY_1^H \right] & {\rm Tr}  \left[ \bY_2 \bY_2^H \right]
	\end{bmatrix}\ .
	\eeq
If $\bC_\bH=\bI_L$, the ML estimators \er{3MLEMP} can be expressed in closed-form as
	\bea
	\hat{F}_{ML} & = & \frac{T_{22}-T_{11}+\sqrt{(T_{22}-T_{11})^2+4|T_{12}|^2}}{2T_{12}} \ , \label{FMLiidMP} \\
\hat{\sigma}_{h}^2 & = & \frac{|\hat{F}_{ML}|^2}{|\hat{F}_{ML}|^2+1} \max \left\{ T_{11}+T_{12}\hat{F}_{ML} -\sigma^2 , 0 \right\} \label{sigmahiidMP} \ , 
	\eea
	provided $T_{12}\neq 0$. $\hfill\diamond$
\end{corollary}

\begin{proof} In Theorem~\ref{3ThrmMLMP}, the ML estimators are given in terms of $\hat{\mu}$ and $\hat{\be}_1$, which jointly maximize the function
\beq
f( \mu, \be_1 ; \sigma^2) \define \be_1^H \bS( \mu ) \be_1 -\sigma^2\sum_{k=1}^L \ln(\mu \lambda_k + \sigma^2) \ . \label{3fdef}
\eeq
For $\bC_\bH = \bI_L$, we have $\lambda_1 = \cdots \lambda_L = 1$, so
\bea
	\bS ( \mu ) \ = \ \frac{\mu}{\mu+ \sigma^2} \bT \ ,
\eea
and $f$ reduces to
\[
f( \mu, \be_1 ; \sigma^2) = L \left[ \frac{\mu}{\mu+ \sigma^2}\be_1^H \bT \be_1 -\sigma^2 \ln(\mu + \sigma^2) \right] \ .
\]
The maximum of $\be_1^H \bT \be_1$ over $\be_1$ is clearly the largest eigenvalue of $\bT$, say $\eta_1$, and is achieved when $\be_1$ is any associated eigenvector.
By direct differentiation, we observe
\[
\frac{\mu \eta_1}{\mu+ \sigma^2} -\sigma^2 \ln(\mu + \sigma^2) \ ,
\]
is maximized by
\beq
\hat{\mu} \ = \  \max \left\{ \eta_1 - \sigma^2, 0\right\} \ . \label{hatmuiidMP}
\eeq

Similar to proofs in Corollary~\ref{3CFML}, we  derive
closed-form formulas for $\eta_1$ and $\be_1$
\bea
\eta_1 & = & \frac{T_{22}+T_{11} + \sqrt{(T_{11}-T_{22})^2+4|T_{12}|^2}}{2} \ \label{3eta1MP} \\
 \hat{\be}_1 & = & \begin{bmatrix}
		E_1 \\
		E_2
	\end{bmatrix} \ = \ \frac{1}{\sqrt{| \eta_1 - T_{11}|^2 + |T_{12}|^2}} \begin{bmatrix}
		T_{12} \\
		\eta_1 - T_{11}
	\end{bmatrix} \ . \label{3be1MP}
\eea
If $T_{12} \neq 0$, substituting $\hat{\mu}$ and $\be_1$ into \er{3MLEMP} yields
\bea
\hat{F}_{ML} \ = \ \frac{\eta_1 - T_{11}}{T_{12}} \ , \ \hat{\sigma}_h^2 \ = \ \frac{| \eta_1 - T_{11}|^2}{| \eta_1 - T_{11}|^2 + |T_{12}|^2} \max \left\{ \eta_1 -\sigma^2 , 0 \right\} \ , \label{3estimatorsMP}
\eea
Substituting \er{3eta1MP} into the first equation gives \er{FMLiidMP}; substituting $\eta_1 = T_{11}+T_{12}\hat{F}_{ML}$ into the second gives \er{sigmahiidMP}.
\end{proof}

\subsection{Multi-Packet Cram\'er-Rao bound}
The entries of Fisher information matrix (FIM) have been derived, for $1\leq i, j\leq 2$, using \cite[pg. 529]{kay} and extension of (15.60) in Kay \cite[pg. 531]{kay},
\beq
[\mathcal{I}(\btt)]_{ij} \ = \ N\cdot \sum_{m=1}^{L}\Trace\left[\bC_k^{-1}\frac{\partial \bC_k}{\partial \theta_i^*}\bC_k^{-1}\frac{\partial \bC_k}{\partial \theta_j}\right] \ ,
\eeq
where $\bC_k$ is given in \er{3Ckdef}. We derive the FIM as
\beq\label{FIM}
\boldsymbol{\mathcal{I}}(\btt) \ = \ N(1+|F|^2)\sum_{k=1}^{L}\frac{\lambda_k^2}{\left[\lambda_k\sigma_h^2(1+|F|^2) + \sigma^2\right]^2}\begin{bmatrix}
  (\sigma_h^2)^2\left(\frac{\lambda_k\sigma_h^2}{\sigma^2}+1\right)& F\sigma_h^2 \\
  F^*\sigma_h^2& 1+|F|^2
\end{bmatrix} \ .
\eeq
For any unbiased estimators $\hat{\btt}$, the classical Cram\'er-Rao bound (CRB) is then calculated as the inverse of FIM,
\beq\label{CRB}
E\left[\left(\hat{\btt}-\btt\right)\left(\hat{\btt}-\btt\right)^H \right]\ \geq \ \mathcal{C}(\btt) \ = \ \boldsymbol{\mathcal{I}}^{-1}(\btt) \ .
\eeq
It appears challenging to derive the CRB in closed-form for general channel correlation $\bC_\bh$.
However, we next investigate a special case, i.e.,  i.i.d. fading channels, where both the maximum-likelihood (ML) estimators and the CRB are in closed-form.

If the temporal correlation matrix $\bC_\bh=\bI_L$,  the multi-packet ML estimator for $F$ is a trivial extension of the single-packet ones, which are given in \er{FMLiid}. The multi-packet CRB under i.i.d. fading is the single-packet CRB in \er{CRB_sp_ch3} scaled by $1/L$, i.e.,
\beq\label{CRB_iid}
\mathcal{C}_F(\btt) \ \define \ \frac{\sigma^2 \sigma_h^2 (1+|F|^2) + \sigma^4}{NL \sigma_h^4} \ , ~~~~
\mathcal{C}_{\sigma_h^2}(\btt) \ \define \ \frac{\left( \sigma_h^2 (1+|F|^2)+\sigma^2 \right)\left( \sigma_h^2 +\sigma^2 \right)}{NL(1+|F|^2)} \ .
\eeq

\subsection{Method of Moments Estimator}

Now we give another estimator for $F$ based on method of moments,
\bea\label{eq_MM}
\hat{F}_{MM} & = & \frac{T_{22}-T_{11}+\sqrt{(T_{22}-T_{11})^2+4|T_{12}|^2}}{2T_{12}} \ ,  
\eea
provided $T_{12}\neq 0$ and $T_{ij}$'s are defined in \eqref{3Tdef}. This is identical to \eqref{FMLiidMP}. 
It can be shown that this MM estimator is also a true ML estimator by treating both $\bH$ and $F$ as deterministic\footnote{Although $\hat{F}_{MM}$ in \er{eq_MM} is another ML estimator \cite[Ch.~3]{wu3}, we call it the MM (method of moments) estimator to distinguish it from the MLE found in Theorem \ref{3ThrmMLMP}. }. One observes that $\hat{F}_{MM}$ remains the same regardless the channel correlation. This renders it a fast algorithm to estimate $F$ in real-time under any fading condition.

\subsection{MMSE Channel Estimator and Ergodic Capacity}
To this end, we have focused on the first step, i.e., joint maximum-likelihood estimation of $F$ and $\sigma_h^2$. We briefly cover the second step, MMSE estimator of $\bH$. If $F$ is known, the MMSE estimator of $\bH$ can be readily derived as, 
\beq\label{eq_Hmmse}
\hat{\bH}_{MMSE} (F) \ = \ \left[\left(1+\left|F\right|^2\right)\bC_\bH+\frac{\sigma_h^2}{\sigma^2}\bI_L\right]^{-1}\bC_\bH \left(\bY_1+F^*\bY_2\right) \ , 
\eeq 
where the parentheses indicate that $\hat{\bH}_{MMSE}$ depends on $F$. The Bayesian CRB for channel estimation with known $F$ can be calculated as
\beq\label{eq_BCRB}
\Trace\left[\left(\sigma_h^2\frac{1+|F|^2}{\sigma^2}\bC_\bH+\bI_L\right)^{-1}\bC_\bH\right]/L \ .
\eeq 

Clearly, one can plug estimators of $F$ into \er{eq_Hmmse} to estimate the channel matrix and measure their efficiency against the Bayesian CRB in \er{eq_BCRB}.

With an estimate of $Z_A$, via estimate of $F$,   the load impedance adapts to achieve conjugate matching,
\beq\label{zL_hat}
\hat{Z}_L \ \define \ \hat{Z}_A^*\ , ~~~~\hat{Z}_{A} \ = \ \frac{Z_2\sqrt{R_1/R_2} \cdot \hat{F} - Z_1}{1 - \sqrt{R_1/R_2}\cdot\hat{F}} \ ,
\eeq
and $\hat{Z}_A$ is calculated via the  invariance principle of MLE (or MM).

Starting from the next packet with \er{zL_hat}, the receiver could perform the minimum mean-square error (MMSE) estimator for channel estimation \cite{bigu}.
A lower bound on ergodic capacity has been derived, when the MMSE estimate is treated as correct during data transmission for one packet \cite[eq.~21]{Hassibi}, i.e.,
\beq
C_l \ = \ E \left[\log_2\left(1 +\gamma_\text{eff} \frac{\bh^H\bh}{N} \right)\right]  \ ,
\eeq
where the distribution is over $\bh\sim\mccn(\bzro,\bI_N)$ and the effective SNR is derived as,
\beq
\gamma_\text{eff} \ \define \  \frac{P(\sigma_h^2-J_{1})}{\sigma_{n}^2+P\cdot J_{1}} \ = \ \gamma \cdot  \frac{1}{1+(1+1/\gamma)N/T} \ ,
\eeq
where $\gamma={P\sigma_h^2}/{\sigma_{n}^2}$ is defined in \er{SNR} and $J_1=\sigma_h^2\cdot \sigma_n^2/(\sigma_n^2+TP\sigma_h^2/N)$. It has been shown that the capacity lower bound in \er{capacity_PCA} can be simplified into a closed-form \cite[eq.~20]{shin},
\beq\label{eq_Cl}
C_l \ = \  \log_2 (e) e^{N/\gamma_\text{eff}}\sum_{k=1}^{N}E_{k}\left(\frac{N}{\gamma_\text{eff}}\right)\ ,
\eeq
where $E_{n}(a)$ is the exponential integral for integer $n\geq 2$ and $\Real\{a\}>0$,
\beq\label{eq_En}
E_{n}(a) \ = \ \int_{1}^{\infty} \frac{e^{-at}}{t^n} d t \ = \ \frac{1}{n-1}\left[e^{-a} -a\cdot E_{n-1}(a)\right] \ .
\eeq
With \eqref{eq_En} and some straightforward algebra, the capacity lower bound \er{eq_Cl} can be written as,
\beq\label{capacity_PCA}
C_l \ = \  \frac{\log_2 e}{2}\left[\left(\frac{N}{\gamma_\text{eff}}-1\right)^2+\frac{8}{3} +\tau\left(\frac{N}{\gamma_\text{eff}}\right) \exp\left(\frac{N}{\gamma_\text{eff}}\right) E_1\left(\frac{N}{\gamma_\text{eff}}\right)\right] \ ,
\eeq
where $N=4$, and $\tau(x)$, a function of $x\in\mbbR^+$, is defined as
\beq
\tau(x) \ \define \ 2 - x\cdot (2+x^2-x) \ .
\eeq

The performance of above estimators and the ergodic capacity are evaluated in the next section.

\section{Numerical Results}\lb{3secV}
In this section, we explore the performance of  estimators in the previous section through numerical examples. Consider a narrow-band MISO communications system with $N=4$ transmit antennas, whose carrier frequency is 2.1 GHz. 
This frequency is chosen based on a down-link operating band in 3GPP E-UTRA \cite{3gpp_TS36101}. The duration of each data packet equals to a sub-frame of NR, i.e., $T_s=1$ ms.
Block-fading channel is assumed, such that during one data packet, the channel remains the same, but it generally varies from packet to packet \cite{bigu}.

For each data packet, a training sequence precedes data sequence \cite[Fig. 1(a)]{liu}.
We take the two partitions of the training sequence $\bX=[\bX_1, \bX_2]$  from a normalized discrete Fourier transform
(DFT) matrix of dimension $K=T/2=32$, e.g., \cite[eq. 10]{bigu}. For example,  $\bX_1$ can be chosen as the first $N$ rows, while $\bX_2$ the next $N$ rows, and $\bX_i\bX_i^H=(KP/N)\bI_N$ for $i=1,2$, where $P$ is the total transmit power for each symbol.  The unknown antenna impedance is that of a dipole $Z_A=73+j42.5 \,\Omega$. The load impedance \er{zL_switch} is $Z_1=50 \,\Omega$ for the first $K=32$ symbols of each training sequence, and $Z_2=60+j20 \,\Omega$ for the remaining $T-K=32$ symbols. From \er{3F}, it follows $F=0.9646 - j0.1032$.

From \er{rho_sym} with $Z_L=Z_1$,  \er{h_pca} and \er{Xs}, we define the average post-detection signal-to-noise ratio (SNR) of a received symbol as
\beq\label{SNR}
\gamma \ \triangleq \ {E\left[\left|\bh^T\bx\right|^2\right]}/{\sigma_n^2 }\ = \ \frac{R_1}{\sigma_{n}^2|Z_A+Z_1|^2}\Trace\left(E\left[\bx\bx^H\right]E\left[\bg^*\bg^T\right]\right) \ = \  \frac{ \sigma_h^2\cdot P}{\sigma_{n}^2} \ .
\eeq
Firstly we  explore  the performance of the true ML estimators  \er{MLE_iid} under a fast fading condition, i.e.,  ${\bC}_\bH = \bI_L$. Note we use fast fading and  i.i.d. Rayleigh fading interchangeably. 

%
%
\begin{figure}
	\centering
	\begin{subfigure}{0.5\textwidth}
		\centering
		\includegraphics[width=1.05\textwidth, keepaspectratio=true]{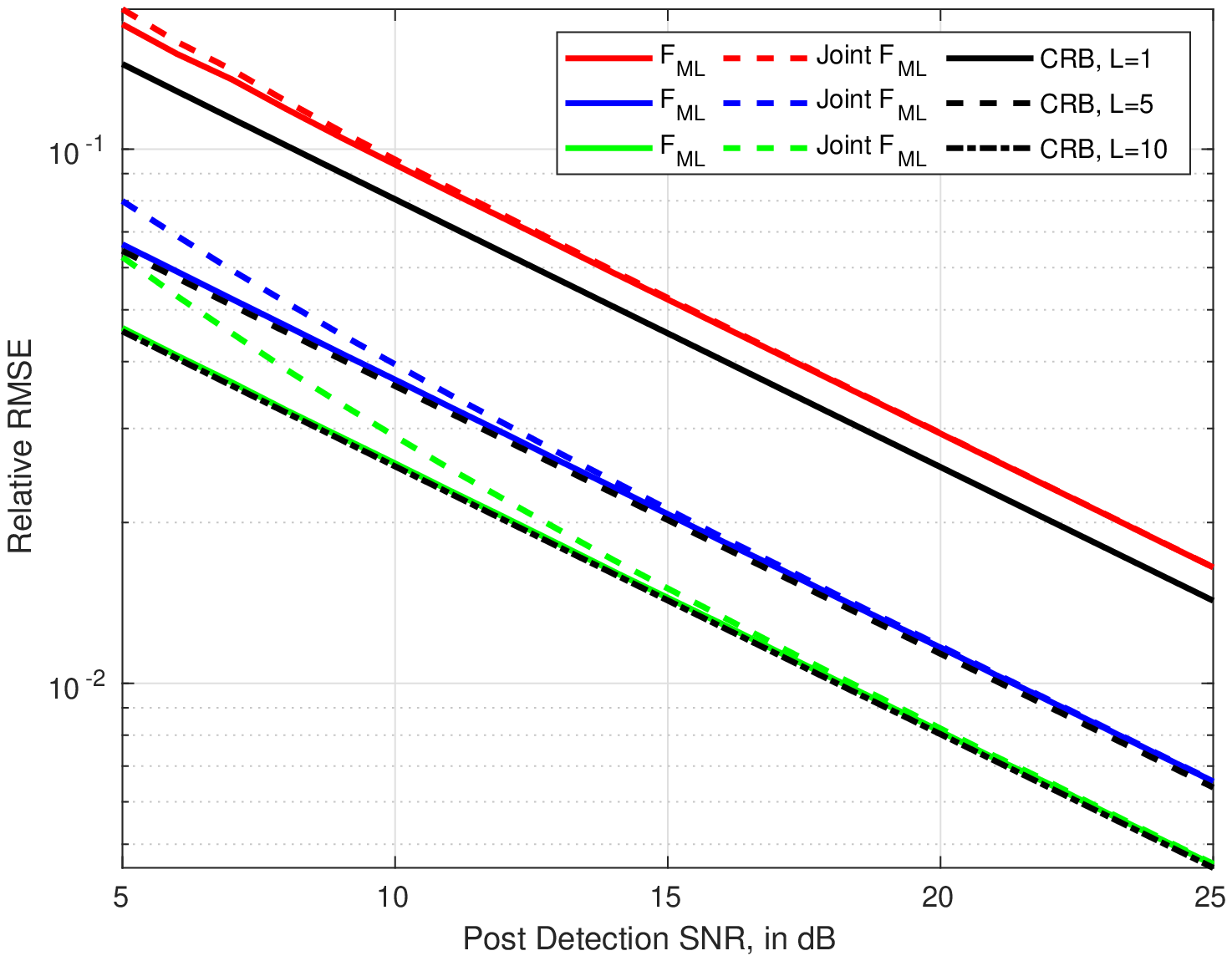}
		\caption{RMSE of Impedance Estimation in i.i.d. Fading.}
		\label{fig_iid_F}
		\vspace{-12pt}
	\end{subfigure}\hfill
	\begin{subfigure}{0.5\textwidth}
		\centering
		\includegraphics[width=1.05\textwidth, keepaspectratio=true]{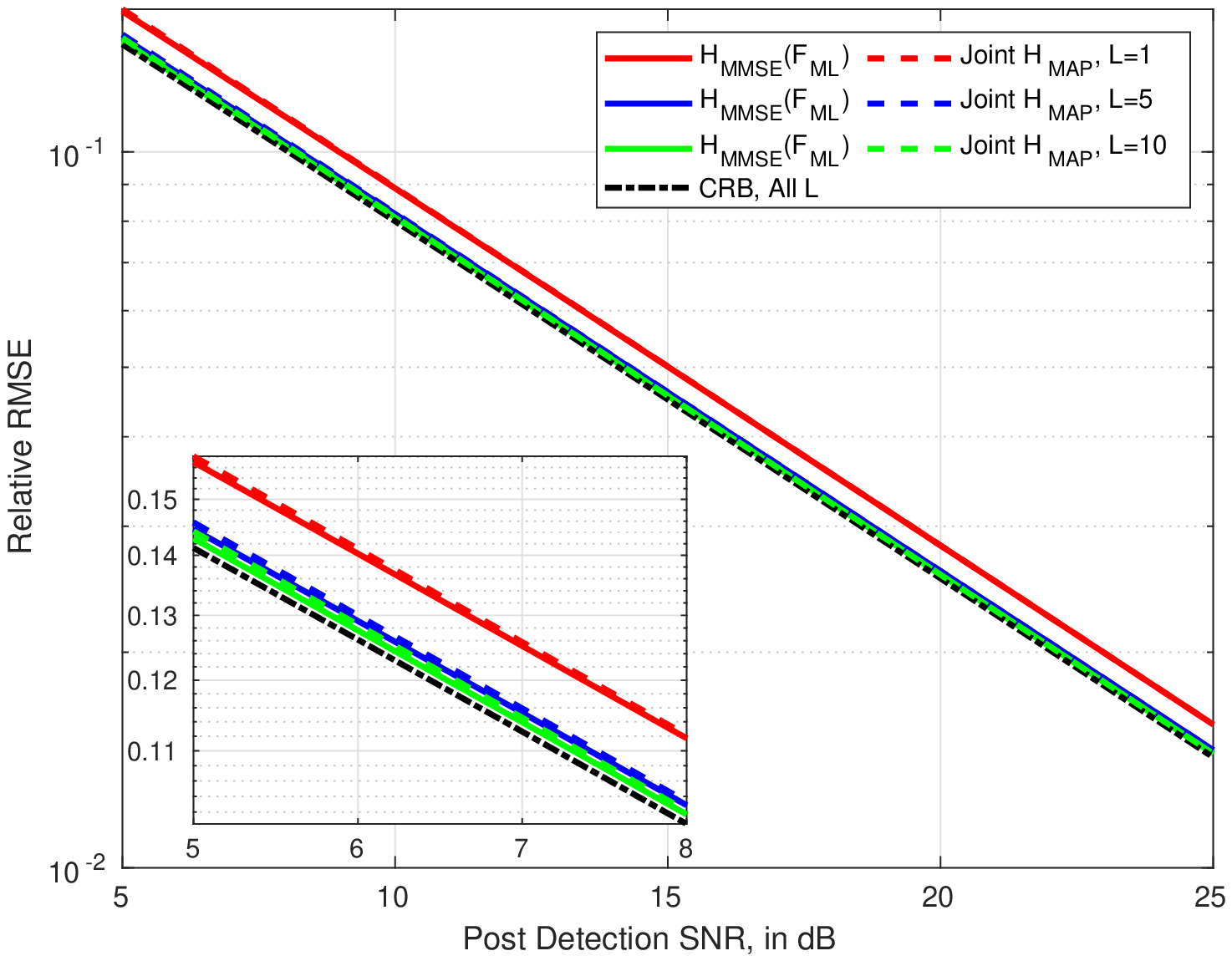}
		\caption{RMSE of Channel Estimation under i.i.d. Fading.}
		\label{fig_iid_sgmH2}
		\vspace{-12pt}
	\end{subfigure}
	\caption{Relative MSE of ML Estimators in i.i.d. Fading with $N=4$.}
	\vspace{-24pt}
\end{figure}

In Fig \ref{fig_iid_F},  the relative root MSE (RMSE) of the ML estimator $\hat{F}_{ML}$ in \er{FMLiid} approaches    its corresponding Cram\'er-Rao bound (CRB)  within a dB for one packet.  This gap vanishes  with a  sufficient number of training packets, e.g., $L=5$. Also plotted is the joint ML estimator derived under the hybrid estimation framework \cite[eq. 29]{wu2}.  At low SNR, $\hat{F}_{ML}$ beats the joint ML estimator $\tilde{F}_{ML}$, despite the latter assumes more knowledge, i.e., knowing $\sigma_h^2$. 

In Fig.  \ref{fig_iid_sgmH2}, we plot the RMSE of the $\hat{\bH}_{MMSE}(\hat{F}_{ML})$ given in \er{eq_Hmmse} and its Bayesian CRB \er{eq_BCRB}. The trend is similar to what we observed for impedance estimation, i.e., as $L$ increases $\hat{\bH}_{MMSE}(\hat{F}_{ML})$ approaches its lower bound. Also note the joint MAP estimator from the hybrid estimation \cite[eq.~20]{wu2} leads to almost identical RMSE, except being slight worse at low SNR. Note this joint MAP estimator shares the exact form as $\hat{\bH}_{MMSE}(\tilde{F}_{ML})$, i.e., using the joint ML estimator of $F$ instead of true ML. The optimality of the true ML $\hat{F}_{ML}$ translates to a superior MMSE channel estimator. 

Note  $\sigma_h^2$ is  a nuisance parameter, which has to be estimated  in this classical estimation framework. But we omit results on its MSE due to page limitation. 
We next explore the performance of  $F$  estimators under correlated fading channels.

Assume Clarke's model for the normalized channel correlation matrix \cite[eq.~2]{badd}. 
The maximum Doppler frequency is
$f_d\define v/\lambda$,
where $v$ is the velocity  of the fastest moving scatterer and $\lambda$ the wave-length of the carrier frequency.
Next we investigate a moderately correlated fading channel, with $v=50$ km/h and $f_d = 97.2$ Hz. The correlation matrix $\bC_\bh$ and its eigenvalues  are, respectively,
\beq
\bC_\bh \ = \ \begin{bmatrix}
	1.0000  &  0.9089  &  0.6602  &  0.3210 &  -0.0199\\
	\vdots 	 &	\ddots   &	\ddots   &	\ddots   &	 \vdots
\end{bmatrix} \ , \nn
\eeq
and $[2.3661\times 10^{-8} ,~7.0552\times 10^{-4},~0.0646,~1.3589,~3.5757]$. All other entries contained in $\bC_\bH$ can be inferred from its first row, because it is a symmetric and Toeplitz matrix.    Under this fading condition, 5 correlated channels provide about 2 orders of temporal diversity. Sequences of these correlated fading path gains are generated using the sum-of-sinusoids model \cite{zhen}.

%
%
\begin{figure}
  \centering
  \begin{subfigure}{0.5\textwidth}
    \centering
    \includegraphics[width=1.05\textwidth, keepaspectratio=true]{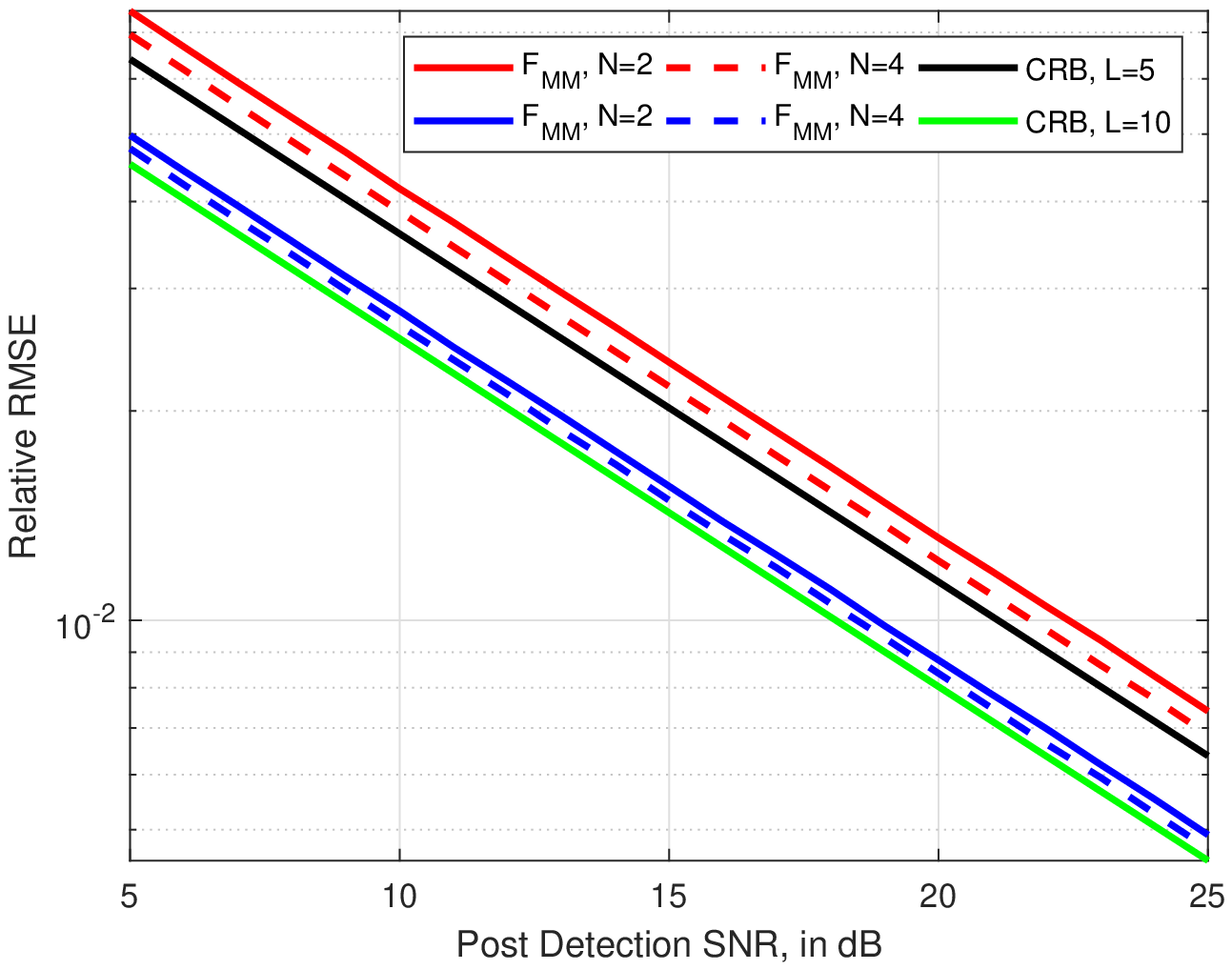}
    \caption{Relative RMSE of $\hat{F}_{MM}$ and CRB.}
    \vspace{-12pt}
    \label{fig_sl_fd_F}
  \end{subfigure}\hfill
  \begin{subfigure}{0.5\textwidth}
    \centering
    \includegraphics[width=1.05\textwidth, keepaspectratio=true]{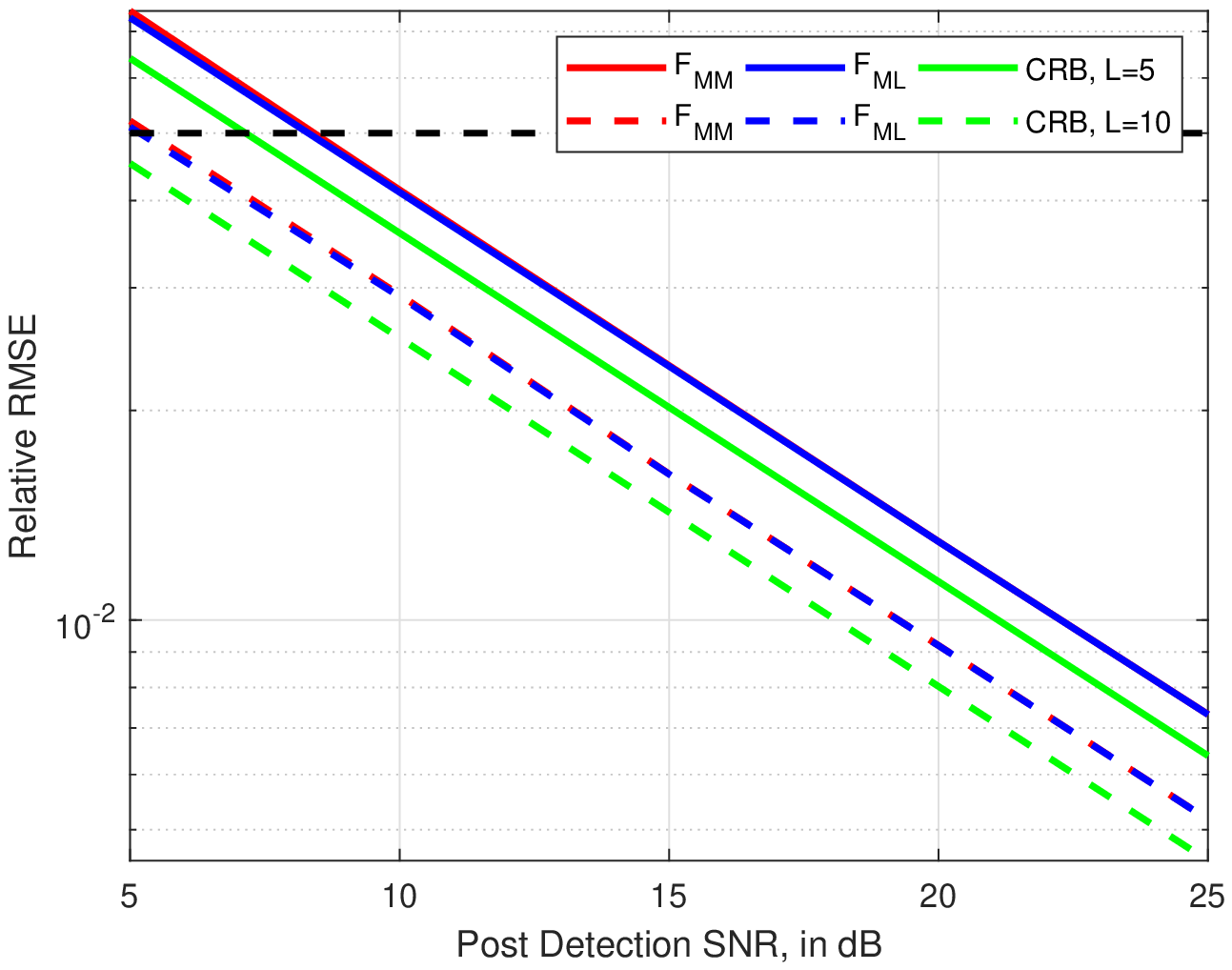}
    \caption{Relative RMSE of $F$ Estimators for $N=4$.}
    \vspace{-12pt}
    \label{fig_sl_fd_Fml_N4}
  \end{subfigure}
  \caption{Diversity's Impact on $\hat{F}_{ML}$ and $\hat{F}_{MM}$.}
  \vspace{-24pt}
\end{figure}

In Figs.~\ref{fig_sl_fd_F}, the MM estimator $\hat{F}_{MM}$ \er{eq_MM} are plotted with four combinations of $N=2,4$ and $L=5,10$.
The general trend is that for fixed $L$ and total transmit power constraint in \er{Xs},  4 transmit antennas render a smaller RMSE  than having only 2. This is due to the independence each extra transmit antenna provides.
The CRB \er{CRB} under  correlated fading generally differs for different $L$ or $N$. However, for the particular cases in Fig.~\ref{fig_sl_fd_F}, this difference is less than 0.3\% at low SNR, and vanishes as SNR increases. So we only plot the $N=2$ CRB's.
The RMSE's of $\hat{F}_{MM}$  is within a fraction of a dB to its CRB for all values of $N$, $L$, and SNR.


When the fading is extremely slow, e.g., $v=5$ km/h and $f_d = 9.72$ Hz, the highly correlated $\bC_\bh$ has eigenvalues 
$[2.172\times 10^{-14},~6.501\times 10^{-10},~6.098\times 10^{-6},~0.0186,~4.981]$.
Five packets only result in one temporal order of diversity. Although most power is concentrated on this single diversity, 
this fading scenario is much worse in terms of impedance estimation efficiency than the moderate fading.

We explore the behavior of $\hat{F}_{ML}$ and $\hat{F}_{MM}$ under slow fading with different packets, $L=5,10$ in Fig.~\ref{fig_sl_fd_Fml_N4}. Here the optimal $\hat{F}_{ML}$ exhibits negligible improvement over the simple $\hat{F}_{MM}$ for all SNR and $L$ considered. 
The 1 dB gap between $\hat{F}_{MM}$ (or $\hat{F}_{ML}$) and the CRB is because  CRB is loose for finite sample size. To this end, another rule of thumbs is that, to be 1 dB within the CRB, a combined 4 orders  of diversity, temporal and/or spatial, is needed. 
Furthermore, we observe in Fig.~\ref{fig_FML_FMM},  $\hat{F}_{ML}$ provides little to no benefit over  $\hat{F}_{MM}$ under all 3 fading conditions, with $L=10$ and $N=4$.  
Then we apply the MMSE estimator for channel estimation by plugging  $\hat{F}_{MM}$ in \eqref{eq_Hmmse}
in Fig.~ \ref{fig_chnlEst_MMSE}. We observe
$\hat{\bH}_{MMSE}(\hat{F})$ with either $F$-estimator closes in the lower bound, where their gaps widen as  temporal correlation increases. 

Note $\hat{F}_{MM}$ can be obtained in closed-form via direct calculation, but $\hat{F}_{ML}$ is generally found via iterative numerical methods, e.g., a line search. Thus, practical systems may choose $\hat{F}_{MM}$ over $\hat{F}_{ML}$ for a better performance-complexity trade-off, which we do next.

%
%
\begin{figure}
  \centering
  \begin{subfigure}{0.5\textwidth}
    \vspace{-12pt}
    \centering
    \includegraphics[width=1.05\textwidth, keepaspectratio=true]{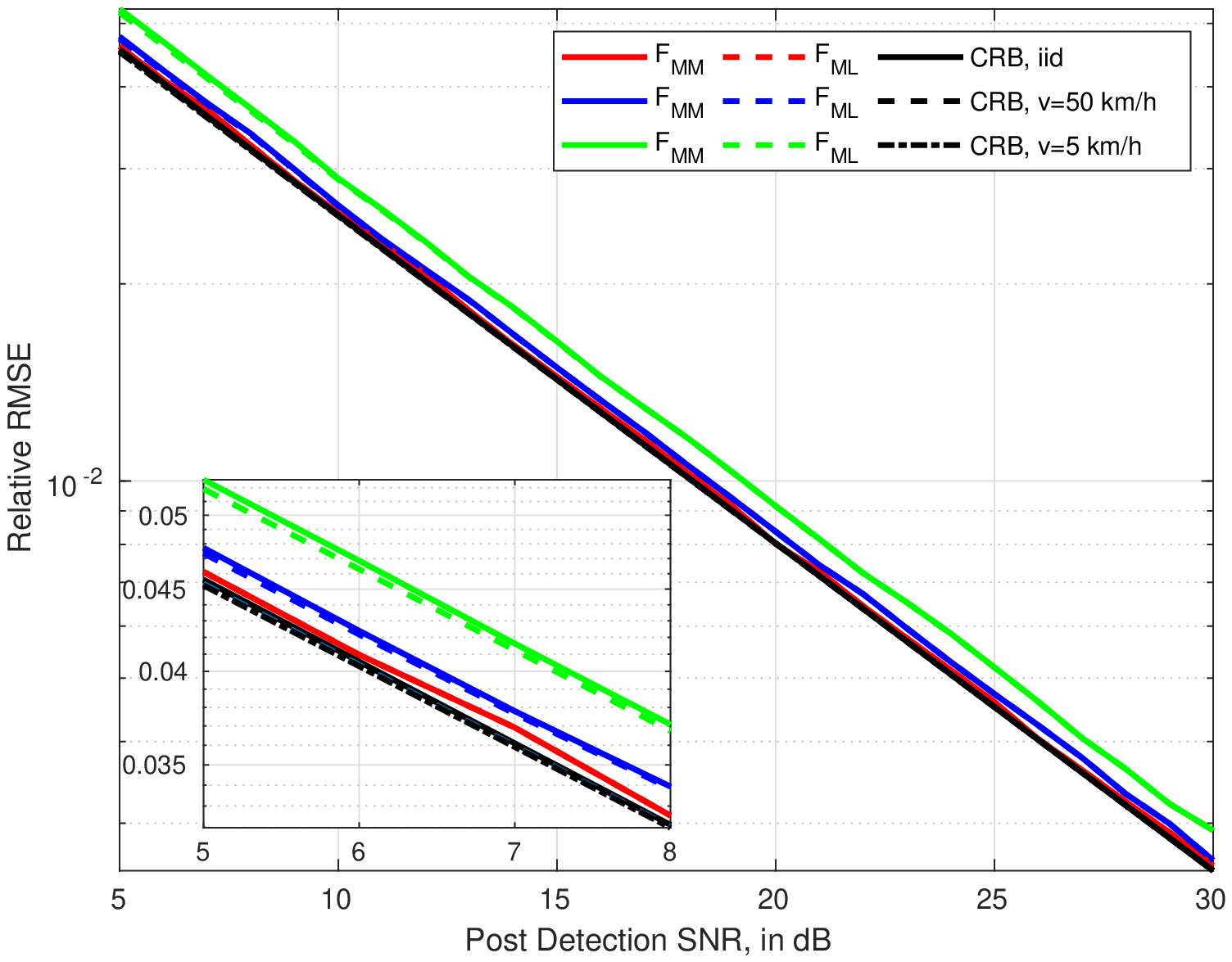}
    \caption{Impedance Estimation under Different Fading.}
    \vspace{-12pt}
    \label{fig_FML_FMM}
  \end{subfigure}\hfill
  \begin{subfigure}{0.5\textwidth}
    \vspace{-12pt}
    \centering
    \includegraphics[width=1.05\textwidth, keepaspectratio=true]{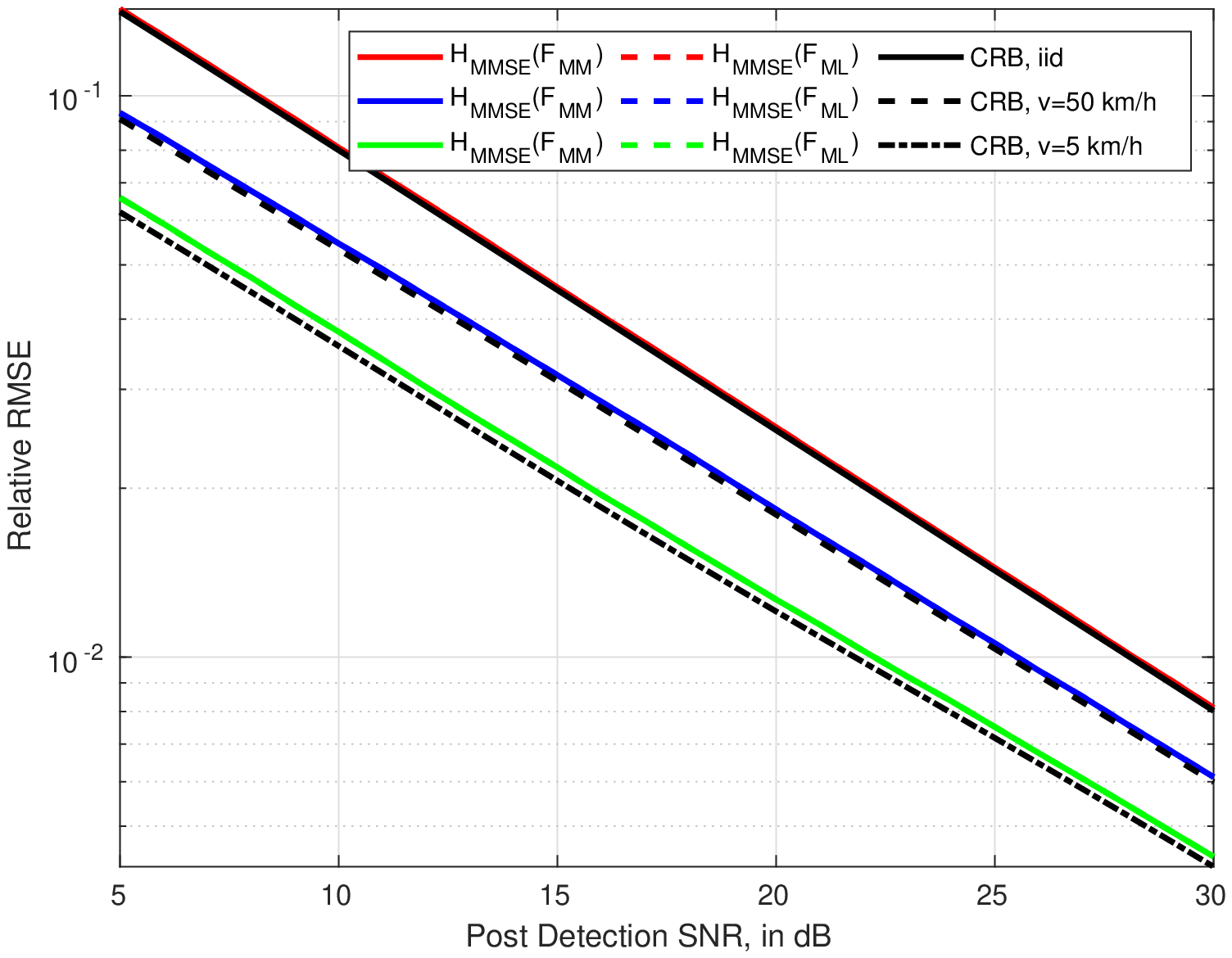}
    \caption{Channel Estimation under Different Fading.}
    \vspace{-12pt}
    \label{fig_chnlEst_MMSE}
  \end{subfigure}
  \caption{Benefits of $\hat{F}_{ML}$ over $\hat{F}_{MM}$ in Various Fading Conditions, $L=10$, $N=4$.}
  \vspace{-24pt}
\end{figure}

Consider two practical systems, whose original load impedance is mismatched, with 5 dB and 3 dB power loss, respectively.  Applying our proposed antenna impedance estimation \er{zL_hat}, the receiver quickly compensates its impedance mismatch and improves (ergodic) capacity, which is shown as $C(\hat{F}_{MM})$ in Figs.~\ref{fig_capacity_5dB} and \ref{fig_capacity_3dB}. The horizontal axis is the  SNR \er{SNR} for the receiver before mismatch compensation.  The solid black line represents  the capacity of the original receiver, while the black dash line is the capacity upper bound, with optimal impedance matching and no channel estimation errors.  One observes that $C(\hat{F}_{MM})$ hones in this capacity upper bound for all SNR and fading conditions plotted. Compared to the original loading condition, the capacity almost doubles at low SNR and gains about 20\% at high SNR for the 5 dB power loss case.

%
%
\begin{figure}
  \centering
  \begin{subfigure}{0.5\textwidth}
    \centering
    \includegraphics[width=1.05\textwidth, keepaspectratio=true]{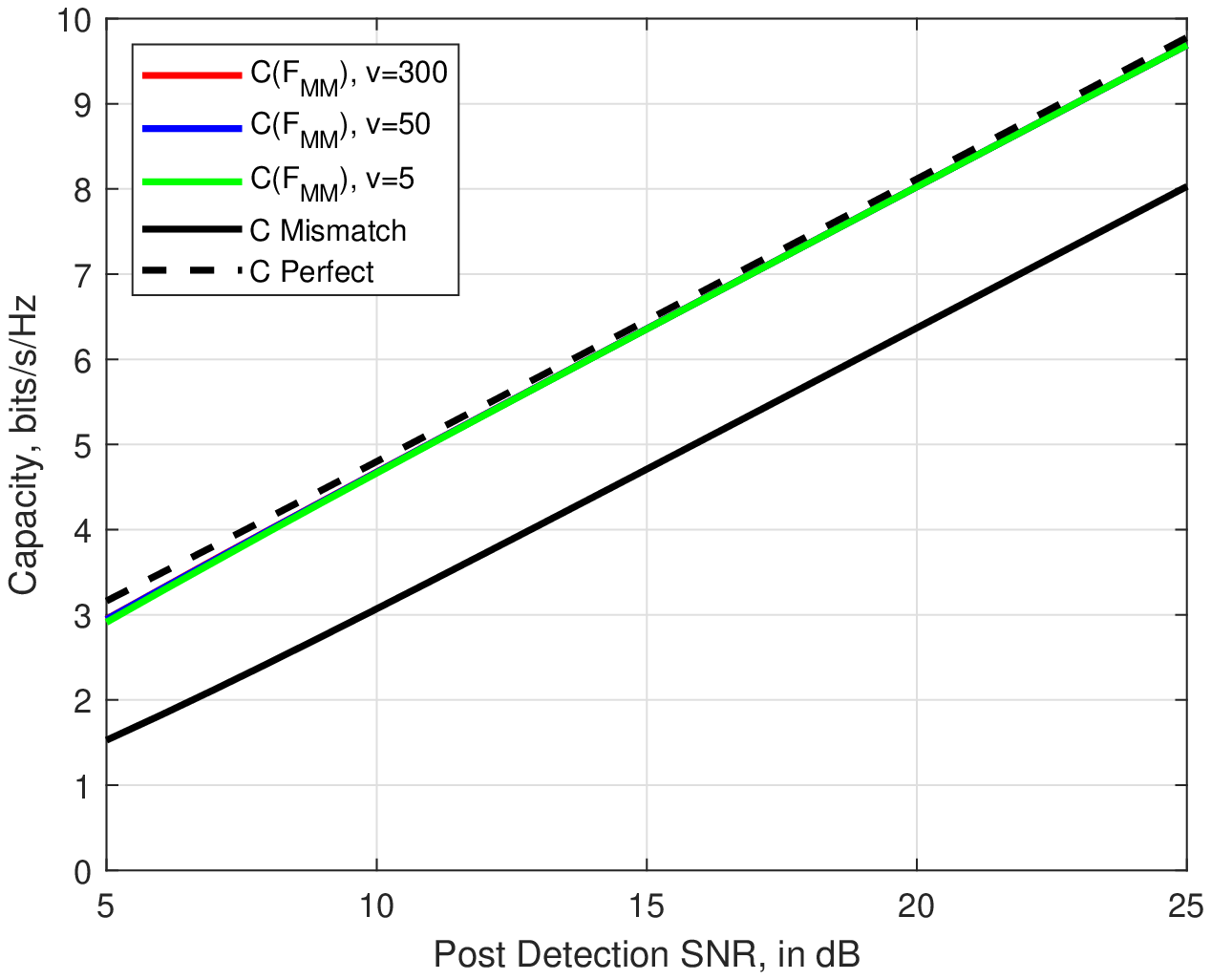}
    \caption{Original Power Loss 5 dB.}
    \vspace{-12pt}
    \label{fig_capacity_5dB}
  \end{subfigure}\hfill
  \begin{subfigure}{0.5\textwidth}
    \centering
    \includegraphics[width=1.05\textwidth, keepaspectratio=true]{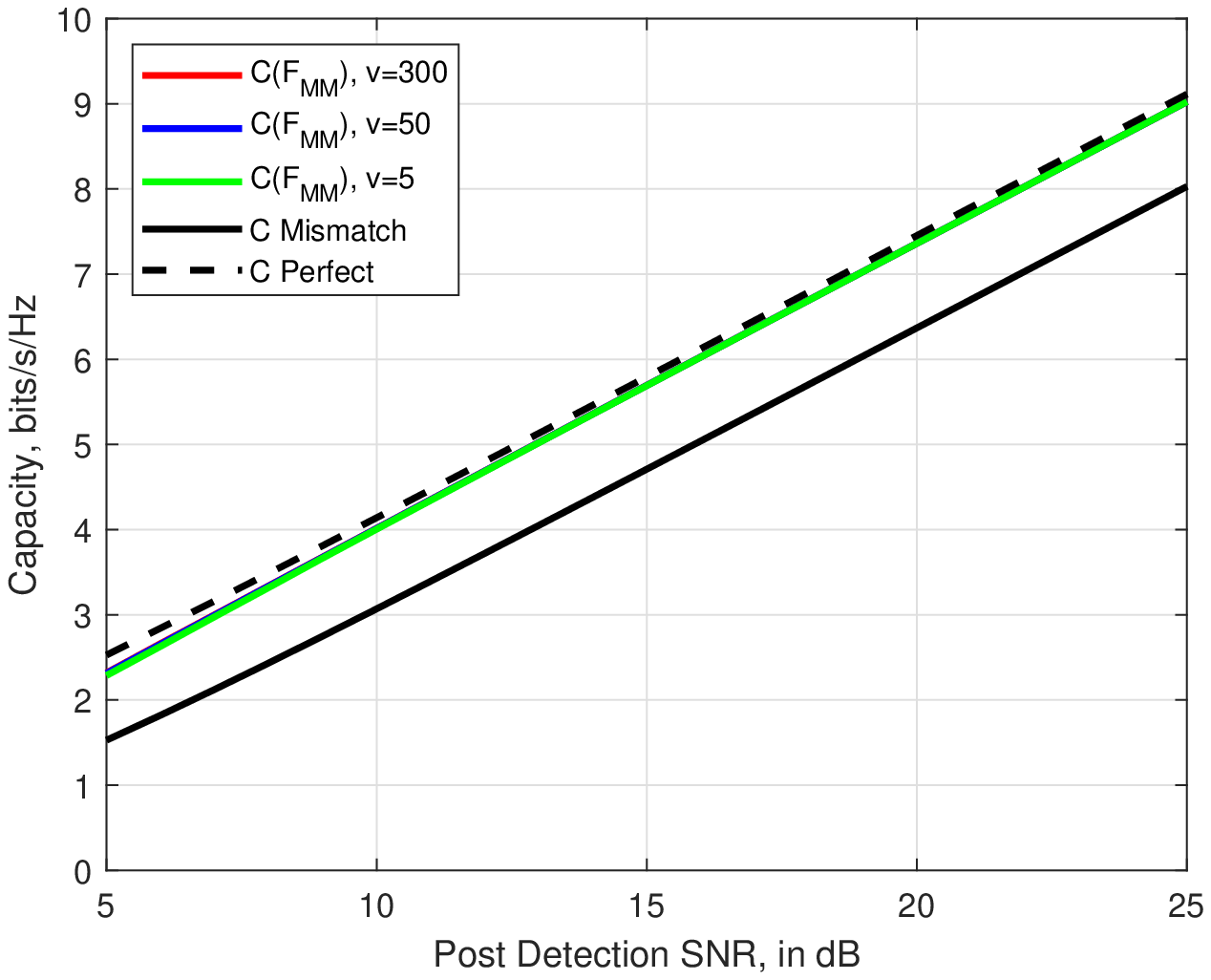}
    \caption{Original Power Loss 3 dB.}
    \vspace{-12pt}
    \label{fig_capacity_3dB}
  \end{subfigure}
  \caption{Ergodic Capacity Evaluation, $L=10$, $N=4$.}
  \vspace{-24pt}
\end{figure}
\section{Conclusion}\lb{3secVI}
In this paper, we formulated the antenna impedance estimation problem at a MISO receiver in classical estimation. We adopted a two-step approach: firstly impedance estimators are derived treating channel path gains as nuisance parameters; then well-known MMSE channel estimators are used with impedance estimates. Specifically, 
we derived the maximum-likelihood (ML) estimator for impedance in closed-form under fast fading. In other fading conditions, we proved this MLE can be found via a scalar optimization. We proposed a simple method of moments (MM) estimator in correlated channels. 

Numerical results demonstrated that both the ML and MM estimators
approach their CRB's given sufficient degrees of diversity, spatial and/or temporal. A rule of thumb is 4 degrees of diversity is needed for a gap within 1 dB to CRB. These findings suggest a fast principal-components based algorithm to estimate antenna impedance in real-time for all Rayleigh fading conditions. This algorithm significantly boosted the ergodic capacity of an originally poorly matched receiver. Thus, our proposed  algorithm suggests accurate impedance estimation and  fast mismatch compensation for future communication systems.

\bibliographystyle{unsrt}

\begin{thebibliography}{40}	
  \bibitem{domi} C.~P.~Domizioli and B.~L.~Hughes, ``Noise correlation in compact diversity receivers,''
  {\em IEEE Trans. Commun.}, vol.~58, no.~5, pp.~1426--1436, May~2010.
  
  \bibitem{domi2} C.~P.~Domizioli and B.~L.~Hughes, ``Front-end design for compact MIMO receivers: A communication theory
  perspective,'' {\em IEEE Trans. Commun.}, vol.~60, no.~10, pp.~2938--2949, Oct.~2012.
  
  \bibitem{gans} M.~J.~Gans, ``Channel capacity between antenna Arrays - Part I: sky noise dominates," {\em IEEE Trans.  Commun.}, vol.~54, no.~9, pp.~1586--1592, Sep.~2006.
  
  \bibitem{gans2} M.~J.~Gans, ``Channel capacity between antenna arrays - Part II: Amplifier noise dominates,'' {\em IEEE Trans. Commun.}, vol.~54, no.~11, pp.~1983--1992, Nov.~2006.
  
  \bibitem{lau} B.~K.~Lau, J.~B.~Andersen, G.~Kristensson and A.~F.~Molisch, ``Impact of Matching Network on Bandwidth of Compact Antenna Arrays,'' {\em IEEE Trans. Antennas Propag.}, vol.~54, no.~11, pp.~3225--3238, Nov.~2006.
  
  \bibitem{wall} J.~W.~Wallace and M.~A.~Jensen, ``Mutual Coupling in MIMO Wireless Systems: A Rigorous Network Theory Analysis,'' {\em IEEE Trans. Wireless Commun.}, vol.~3, no.~4, pp~1317--1325, Jul.~2004
  
  
  
  
  
  \bibitem{moha}
  R.~Mohammadkhani and J.~S.~Thompson, ``Adaptive Uncoupled Termination for Coupled Arrays in MIMO Systems,'' {\em IEEE Trans. Antennas Propag.}, vol.~61, no.~8, pp.~4284--4295,  May~2013.
  
  \bibitem{vasi}
  I. Vasilev, J.~Lindstrand, V.~Plicanic, H.~Sjoland and B.~K.~Lau, ``Experimental Investigation of Adaptive Impedance Matching for a MIMO Terminal With CMOS-SOI Tuners,'' {\em IEEE Trans. Microw. Theory Techn.}, vol.~64, no.~5, pp.~1622--1633, Apr.~2016.
  
  
  \bibitem{ali}
  S.~M.~Ali, M.~Buckley, J.~Deforge, J.~Warden and A.~Danak, ``Dynamic Measurement of Complex Impedance in Real-Time for Smart Handset Applications,'' {\em IEEE Trans. Microw. Theory Techn.}, vol.~61, no.~9, pp.~3453--3460, Aug.~2013.
  
  \bibitem{hass} 
  Y.~Hassan and A.~Wittneben, ``Joint Spatial Channel and Coupling Impedance Matrices Estimation in Compact MIMO Systems : The Use of Adaptive Loads,'' {\em IEEE 26th Annual International Symposium on Personal, Indoor, and Mobile Radio Communications (PIMRC)},  pp.~29--33, 2015.
  
  
  \bibitem{wu}
  S.~Wu and B.~L.~Hughes, ``Training-based joint channel and impedance estimation,'' {\em IEEE 2018 52nd Annual Conference on Information Sciences and Systems (CISS)},  Princeton University, NJ, Mar. 2018. 
  
  
  
  \bibitem{wu2} 
  S.~Wu and B.~L.~Hughes,  ``A Hybrid Approach to Joint Estimation of Channel and Antenna Impedance,'' {\em IEEE  2018 52nd Asilomar Conference on Signals, Systems, and Computers}, Pacific Grove, CA, Oct.~2018. 
  
  \bibitem{wu3} 
  S.~Wu,  ``Joint Antenna Impedance and Channel Estimation at Multiple-input, Multiple-output Receivers,'' Ph.D. paper, North Carolina State University, 2019. 
  
  
  \bibitem{mes09} 
  Y.~Noam and H.~Messer, ``Notes on the Tightness of the Hybrid Cram\'{e}r-Rao Lower Bound,'' {\em IEEE Trans. Signal Process.}, vol.~57, no.~6, pp.~2074--2084, Jun.~2009. 
  
  
  \bibitem{kay} S.~M.~Kay, {\em Fundamentals of Statistical Signal Processing: Estimation Theory}. Upper Saddle River, New Jersey: Prentice Hall, 1993.
  
  \bibitem{hjor}
  A.~Hj{\o}rungnes, {\em Complex-Valued Matrix Derivatives: With Applications in Signal Processing and Communications}. Cambridge: Cambridge University Press, 2011.
  
  
  
%
  \bibitem{brew}
  J.~W.~Brewer, ``Kronecker Products and Matrix Calculus in System Theory,'' {\em IEEE Trans. Circuits Syst.}, vol.~25, no.~9, pp.~772--781, 1978.
  
  \bibitem{zhen} 
  Y.~R.~Zheng and C.~Xiao, ``Simulation models with correct statistical properties for Rayleigh fading channels,'' {\em IEEE Trans. Commun.}, vol.~51, no.~6, pp.~920--928, Jun.~2003.
  
  \bibitem{badd} K.~E.~Baddour and N.~C.~Beaulieu, ``Autoregressive modeling for fading channel simulation,'' {\em IEEE Trans. Wireless Commun.}, vol.~4, no.~4, pp.~1650--1662, Jul.~2005.
  
  
  
  
  \bibitem{bigu} M.~Biguesh and A.~B.~Gershman, ``Training-based MIMO channel estimation: a study of estimator tradeoffs and optimal training signals,'' {\em IEEE Trans. Signal Process.}, vol.~54, no.~3, pp.~884--893, Mar.~2006.
  
  
  \bibitem{liu} 
  Y.~Liu, Z.~Tan, H.~Hu, L.~J.~Cimini, and G.~Y.~Li, ``Channel Estimation for OFDM,'' {\em IEEE Commun. Surv. Tutorials}, vol.~16, no.~4, pp.~1891--1908,~2014.
  
%
  
  
  
  

  
  
  \bibitem{gill}
  J.~Gillard, ``Asymptotic variance–covariance matrices for the linear structural model,'' {\em Stat. Methodol.}, vol.~8, no.~3, pp.~291--303, May~2011.
  
  
  \bibitem{ivrl} 
  M.~T.~Ivrla\v{c} and J.~A.~Nossek,, ``Toward a Circuit Theory of Communication,'' {\em IEEE Trans. Circuits Syst. I Regul. Pap.}, vol.~57, no.~7, pp.~1663--1683, Jul.~2010. 
  
  
  \bibitem{poza} D.~M.~Pozar, {\em Microwave Engineering}. 4th ed.  New York, NY: Wiley, 2011.
  
  
  \bibitem{vu} 
  M.~Vu, ``MISO Capacity with Per-Antenna Power Constraint,'' {\em IEEE Trans. Commun.}, vol.~59, no.~5, pp.~1268--1274, May~2011.
  

  \bibitem{Hassibi} 
  B.~Hassibi and B.~M.~Hochwald, ``How much training is needed in multiple-antenna wireless links?,'' {\em IEEE Trans. Inf. Theory}, vol.~49, no.~4, pp. 951--963, 2003.
  
  \bibitem{shin} 
  H.~Shin and J.~H.~Lee, ``Capacity of multiple-antenna fading channels: spatial fading correlation, double scattering, and keyhole,'' {\em IEEE Trans. Inf. Theory}, vol.~49, no.~10, pp.~2636--2647, Oct.~2003.
  
  \bibitem{lu}
  Q.~Lu, Y.~Bar-Shalom, P.~Willett, F.~Palmieri, and F.~Daum, ``The Multidimensional Cram\'er-Rao-Leibniz Lower Bound for Likelihood Functions With Parameter-Dependent Support,'' {\em IEEE Trans. Aerosp. Electron. Syst.}, vol.~53, no.~5, pp.~2331--2343, Oct.~2017.
  
  \bibitem{wies}
  A.~Wiesel, Y.~C.~Eldar, and A.~Yeredor, ``Linear Regression With Gaussian Model Uncertainty: Algorithms and Bounds,'' {\em IEEE Trans. Signal Process.}, vol. 56, no. 6, pp. 2194--2205, Jun. 2008.
  
  %
    \bibitem{3gpp_TS36101}
    3GPP TS 36.101: ``Evolved Universal Terrestrial Radio Access (E-UTRA); User Equipment (UE) radio transmission and reception''.
  
  
\end{thebibliography}


\ifcomments
For a general non-singular temporal correlation matrix $\bC_\bH$, the ML estimators  \er{3MLEMP} can easily be
approximated under high SNR conditions, such that $\sigma^2 \approx 0$. In Theorem~\ref{3ThrmMLMP}, the ML estimators are given in terms of $\hat{\mu}$ and $\hat{\be}_1$, which jointly maximize the function $f( \mu, \be_1 ; \sigma^2)$ in \er{3fdef}. If $\bC_\bH$ is any non-singular matrix,
For small $\sigma^2$, we can expand this function as
\bea
f( \mu, \be_1 ; \sigma^2) & = & \sum_{k=1}^L \left[ \frac{\mu \lambda_k}{\mu \lambda_k + \sigma^2} \be_1^H\bS_k\be_1 -\sigma^2 \ln(\mu \lambda_k + \sigma^2) \right] \nn \\
& = & \sum_{k=1}^L \left[ \left( 1 - \frac{\sigma^2}{\mu \lambda_k} \right) \be_1^H\bS_k\be_1 - \sigma^2 \ln(\mu \lambda_k) \right] + \mathcal{O}( \sigma^4 ) \nn \\
& = & L g( \mu, \be_1 ; \sigma^2) - \sigma^2\ln {\rm det}[ \bC_\bH]  + \mathcal{O}( \sigma^4 ) \ ,
\eea
where
\bea
g( \mu, \be ; \sigma^2) & \define &  \be^H \left[ \bT  - \frac{\sigma^2}{\mu} \bS \right] \be - \sigma^2 \ln(\mu ) \ , \label{3gdefMP}
\eea
where $\bT$ is given in \er{3Tdef} and
\beq\label{3Sdef}
\bS \ \define \ \frac{1}{L} \begin{bmatrix}
  {\rm Tr}  \left[ \bC_\bH^{-1}\bY_1 \bY_1^H \right] & {\rm Tr}  \left[ \bC_\bH^{-1}\bY_1 \bY_2^H \right] \\
  {\rm Tr}  \left[ \bC_\bH^{-1}\bY_2 \bY_1^H \right] & {\rm Tr}  \left[ \bC_\bH^{-1}\bY_2 \bY_2^H \right]
\end{bmatrix}\ .
\eeq
We define the approximate ``low-noise ML estimators'' to be the estimators  \er{3estimatorsMP} with $\hat{\mu}$ and $\hat{\be}_1$ replaced by the $\be, \mu$ that jointly maximize $g$. The following result provides a simple characterization of these estimators.

\begin{theorem}[Low-Noise ML Estimators]\label{3LNML} Let $\bY_1$ and $\bY_2$ be the sufficient statistics in \er{ysMP}, where $F$ and $\sigma_h^2$ are unknown constants. Suppose $\mu_o > 0$ and $\be_o \in \mathbb{C}^2$ satisfy the following two conditions: (1) $\be_o$ is a unit eigenvector of $\bP \define \bT - (\sigma^2/\mu_o) \bS$ corresponding to its largest eigenvalue, and (2) $\mu_o= \be_o^H \bS \be_o$. Then $\mu_o,\be_o$ are the global maxima of $g( \mu, \be ; \sigma^2)$, and the low-noise ML estimators are given by
  \bea
  \hat{F}_{ML} & = & \frac{P_{22}-P_{11}+\sqrt{(P_{22}-P_{11})^2+4|P_{12}|^2}}{2P_{12}} \ ,  \\
  \hat{\sigma}_{h}^2 & = & \frac{|\hat{F}_{ML}|^2}{|\hat{F}_{ML}|^2+1} \max \left\{ P_{11}+P_{12}\hat{F}_{ML} -\sigma^2 , 0 \right\}  \ , 
  \eea
  provided $P_{12}\neq 0$. Moreover, the value of $\mu_o$ is unique. $\hfill\diamond$
\end{theorem}

\begin{proof} If $\mu$ and $\be$ maximize
  $g( \mu, \be ; \sigma^2)$ subject to the constraint $\be^H \be = 1$, they must be a critical point of the Lagrangian $L( \mu, \be ; \sigma^2)=g( \mu, \be ; \sigma^2)- \eta \be^H \be$ for some real $\eta$. It follows $\mu$ and $\be$ satisfy
  \bea
  {\bf 0}=\nabla_\be L & = & 2 \be^H \left[ \bT - (\sigma^2/\mu) \bS \right] - 2 \eta \be^H \nn \\
  0 = \frac{\partial L}{\partial \mu} & = & - \frac{\sigma^2}{\mu^2} \be^H \bS \be + \frac{\sigma^2}{\mu} \ ,
  \eea
  for some $\eta$. For a given $\mu$, the first equation is satisfied only when $\eta$ is an eigenvalue of $\bT - (\sigma^2/\mu) \bS$ and $\be$ is a corresponding unit eigenvector. From \er{3gdefMP}, it is $\eta$ must be the largest eigenvalue of $\bT - (\sigma^2/\mu) \bS$, since otherwise $\be$ would not achieve $\max_{\be: \be^H \be =1} g( \be, \mu; \sigma^2)$. Similarly, for a given $\be$, the unique solution of the second equation is $\mu= \be^H \bS \be$.
  
  We now prove $\mu_o, \be_o$ is the global maxima of $g( \be , \mu ; \sigma^2 )$. For any $\mu > 0$, let $\eta_\mu $ denote the largest eigenvalue of $\bT - \frac{\sigma^2}{\mu} \bS$, and $\be_\mu$ be any associated unit eigenvector. We claim $\mu > \mu^\prime$ implies
  \bea
  \be_\mu^H \bS \be_\mu  \ \geq \ \be_{\mu^\prime}^H \bS \be_{\mu^\prime} \ . \label{3inequalities}
  \eea
  To see this, note the Rayleigh-Ritz Theorem implies
  \bea
  \eta_{\mu^\prime} & = & \be_{\mu^\prime}^H \left[ \bT - \frac{\sigma^2}{\mu^\prime} \bS \right] \be_{\mu^\prime} \nn \ \geq \
  \be_{\mu}^H \left[ \bT - \frac{\sigma^2}{\mu^\prime} \bS \right] \be_\mu  \ .
  \eea
  Since
  \bea
  \bT - \frac{\sigma^2}{\mu^\prime} \bS = \bT - \frac{\sigma^2}{\mu} \bS  + \sigma^2 \left( \frac{1}{\mu}-\frac{1}{\mu^\prime}\right) \bS \ , \label{3matrixidentity}
  \eea
  it follows
  \bea
  && \be_{\mu^\prime} \left[ \bT - \frac{\sigma^2}{\mu} \bS \right] \be_{\mu^\prime} + \sigma^2 \left( \frac{1}{\mu}-\frac{1}{\mu^\prime}\right) \be_{\mu^\prime}^H \bS \be_{\mu^\prime} \nn \\
  & \geq & \be_\mu^H \left[ \bT - \frac{\sigma^2}{\mu} \bS \right] \be_\mu + \sigma^2 \left( \frac{1}{\mu}-\frac{1}{\mu^\prime}\right) \be_\mu^H \bS \be_\mu \nn \\
  & = & \eta ( \mu ) + \sigma^2 \left( \frac{1}{\mu}-\frac{1}{\mu^\prime}\right) \be_\mu^H \bS \be_\mu \ ,
  \eea
  and hence from the Rayleigh-Ritz Theorem
  \bea
  \sigma^2 \left( \frac{1}{\mu^\prime}-\frac{1}{\mu}\right) \left[ \be_{\mu^\prime}^H \bS \be_{\mu^\prime} - \be_{\mu}^H \bS \be_{\mu} \right]
  \ \geq \ \eta ( \mu ) -\be_{\mu^\prime}^H \left[ \bT - \frac{\sigma^2}{\mu} \bS \right] \be_{\mu^\prime} \ \geq \ 0 \ ,
  \eea
  thereby proving the claim.
  
  Now let $\be_o, \mu_o$ satisfy the necessary conditions. For all $\mu > 0$
  \bea
  \eta ( \mu ) & = & \be_\mu^H \left[ \bT - \frac{\sigma^2}{\mu} \bS \right] \be_\mu \nn \\
  & = & \be_\mu^H \left[ \bT - \frac{\sigma^2}{\mu_o} \bS \right]\be_\mu + \sigma^2 \left( \frac{1}{\mu_o}-\frac{1}{\mu}\right) \be_\mu^H \bS \be_\mu \nn \\
  & \leq &  \eta ( \mu_o ) + \sigma^2 \left( \frac{1}{\mu_o}-\frac{1}{\mu}\right) \be_\mu^H \bS \be_\mu \nn \\
  & = &  \eta ( \mu_o ) + \sigma^2 \left( 1-\frac{\mu_o}{\mu}\right) \theta \ , \
  \eea
  where the first inequality follows from the Rayleigh-Ritz Theorem, and the second from \er{3inequalities}. From  \er{3inequalities}, $\theta \geq 1$ if $\mu \geq \mu_o$.
  
  Similarly, this inequality also holds for $\mu < \mu_o$, since
  the coefficient of $\be_\mu^H \bS \be_\mu $ above is then negative. Thus, for all
  $\be$ and $\mu$, we have 
  \bea
  g( \be, \mu ; \sigma^2 ) & \define &  \be^H \left[ \bT  - \frac{\sigma^2}{\mu} \bS \right] \be - \sigma^2 \ln(\mu ) \nn \\
  & \leq &  \eta_\mu - \sigma^2 \ln(\mu ) \nn \\
  & \leq & \eta_{\mu_o} + \sigma^2 \left( 1-\frac{\mu_o}{\mu}\right) - \sigma^2 \ln(\mu ) \nn \\
  & = & g( \be_o, \mu_o ; \sigma^2 ) + \sigma^2 \left[ \left( 1-\frac{\mu_o}{\mu}\right) + \ln \left(\frac{\mu_o}{\mu} \right) \right] \ .
  \eea
  Recall $\ln ( 1+x) \leq x$ for all $x > -1$, with equality if and only if $x=1$. For $x = \mu_o/\mu -1$, this implies the bracketed term is negative for all
  $\mu \neq \mu_o$. Thus, $\be_o, \mu_o$ achieves the global maximum of $g$ and 
  the value of $\mu_o$ is unique. 
\end{proof}
\fi

\ifcomments
Can we prove that the log-likelihood function (LLF) is concave in $\mu$, such that it can be found in close-form? 

\bea\label{eq_new_dev}
&&\ln p( \bY_1 , \bY_2 ; \btt) \nn\\
& = & \sum_{k=1}^L \ln p\left(\bw_{k1} , \bw_{k2} ;\btt\right) \nn \\
&=& B + \frac{N}{\sigma^2} \sum_{k=1}^L \left[ \frac{\mu \lambda_k}{\mu \lambda_k + \sigma^2} \be_1^H\bS_k\be_1 -\sigma^2\ln(\mu \lambda_k + \sigma^2) \right] \nn \\
&=& B + \frac{N}{\sigma^2} \sum_{k=1}^L \left[ \be_1^H\bS_k\be_1 -\frac{\sigma^2}{\mu \lambda_k + \sigma^2} \be_1^H\bS_k\be_1 -\sigma^2\ln(\mu \lambda_k + \sigma^2) \right] \nn \\
&=& B+ \frac{N}{\sigma^2} \sum_{k=1}^L\be_1^H\bS_k\be_1 - N\sum_{k=1}^L \left[ \frac{1}{\mu \lambda_k + \sigma^2} \be_1^H\bS_k\be_1 +\ln(\mu \lambda_k + \sigma^2) \right]\nn\\
&=& 
\eea
where $B$ and $C$ do not depend on the parameters and
\fi

\end{document}